\documentclass[conference]{IEEEtran}
\PassOptionsToPackage{dvipsnames,table}{xcolor}  
\ifCLASSOPTIONcompsoc
  \usepackage[nocompress]{cite}
\else
  \usepackage{cite}
\fi

\usepackage{amsmath,amssymb,amsfonts}
\usepackage[ruled,vlined]{algorithm2e}
\usepackage[final]{changes}
\usepackage{multirow}
\usepackage{pifont}
\usepackage{threeparttable}
\usepackage{wasysym}
\usepackage{graphicx}
\usepackage{subcaption}
\usepackage{booktabs}

\usepackage{amsthm}
\newtheorem{theorem}{Theorem}
\begin{document}
%
\title{{\textsc{ SecNeuron:}} Reliable and Flexible Abuse Control in Local LLMs via Hybrid Neuron Encryption}

\author{
\IEEEauthorblockN{
Zhiqiang Wang\IEEEauthorrefmark{1},
Haohua Du\IEEEauthorrefmark{2},
Junyang Wang\IEEEauthorrefmark{1},
Haifeng Sun\IEEEauthorrefmark{1}, 
Kaiwen Guo\IEEEauthorrefmark{3},
Haikuo Yu\IEEEauthorrefmark{1},
Chao Liu\IEEEauthorrefmark{3},
Xiang-Yang Li\IEEEauthorrefmark{1}
}
\IEEEauthorblockA{
\IEEEauthorrefmark{1}University of Science and Technology of China \\
Email: \{sa21221041, iswangjy, sun1998, yhk7786\}@mail.ustc.edu.cn, xiangyangli@ustc.edu.cn \\
\IEEEauthorrefmark{2}Beihang University \\
Email: duhaohua@buaa.edu.cn \\
\IEEEauthorrefmark{3}Ocean University of China \\
Email: \{kevinguo, liuchao\}@ouc.edu.cn
}
}

\maketitle

 \begin{abstract}
\deleted{An increasing number of users have chosen to locally deploy large language models (LLMs), which raises significant security risks associated with potential abuse, such as harmful content generation or intellectual property leakage. Existing methods primarily target cloud-based LLM deployments and face several challenges when applied to local deployments, such as vulnerability to circumvention (malicious I/O filtering), uncertain effectiveness (distillation or safety alignment) and high costs (unlearning).}
\added{Large language models (LLMs) with diverse capabilities are increasingly being deployed in local environments, presenting significant security and controllability challenges. These locally deployed LLMs operate outside the direct control of developers, rendering them more susceptible to abuse. Existing mitigation techniques mainly designed for cloud-based LLM services are frequently circumvented or ineffective in deployer-controlled environments.}

We propose \textsc{SecNeuron}, \deleted{an innovative approach}\added{the first framework that seamlessly embeds classic access control within the intrinsic capabilities of LLMs}, achieving reliable, cost-effective, flexible, and certified abuse control for local deployed LLMs.
\textsc{SecNeuron} employs neuron-level encryption and selective decryption to dynamically control the task-specific capabilities of LLMs, limiting unauthorized task abuse without compromising others. 
We first design a task-specific neuron extraction mechanism to decouple logically related neurons and construct a \deleted{cross-linguistic alignment}\added{layered} policy tree for handling coupled neurons. 
We then introduce a flexible and efficient hybrid encryption framework for millions of neurons in LLMs.
Finally, we developed a distribution-based decrypted neuron detection mechanism on ciphertext to ensure the effectiveness of partially decrypted LLMs.
We proved that \textsc{SecNeuron} satisfies IND-CPA Security and Collusion Resistance Security under the \deleted{Neuron Isolation Principle}\added{Task Controllability Principle}. Experiments on various task settings show that \textsc{SecNeuron} limits unauthorized task accuracy to below 25\% while keeping authorized accuracy loss with 2\%. Using an unauthorized Code task example, the accuracy of abuse-related malicious code generation was reduced from 59\% to 15\%. \textsc{SecNeuron} also mitigates unauthorized data leakage, reducing PII extraction rates to below 5\% and membership inference to random guesses. Additionally, \textsc{SecNeuron} enables one-time encryption \& transmission and multi-party selective decryption, requiring millisecond-level key generation and byte-level key exchange for local LLM capability adjustment.
\end{abstract}


%
\IEEEpeerreviewmaketitle
\begin{figure}[h]
  \centering
  \includegraphics[width=\linewidth]{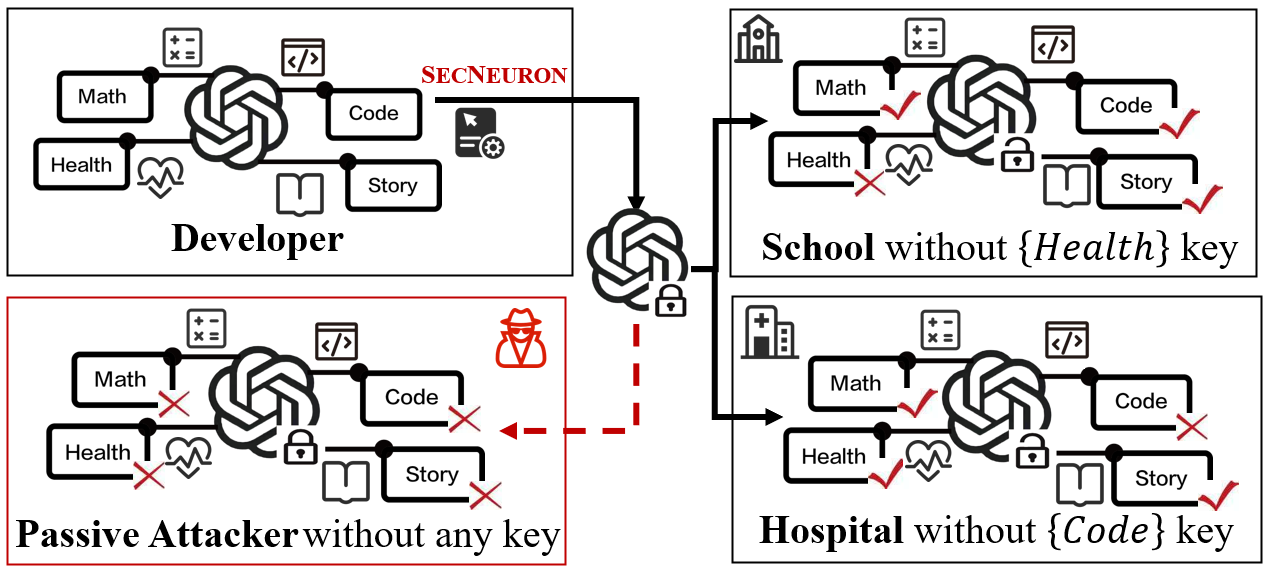}
  \caption{Workflow of \textsc{SecNeuron}. Developer encrypts their LLM once (\textit{One-time Encryption}).
  Different deployers download the same encrypted LLM and dynamically decrypt authorized tasks while restricting unauthorized capabilities to mitigate abuse. (\textit{Multi-party Selective Decryption}). The developer maintained and released one single encrypted LLM, while deployers also download it once even if permissions change (\textit{One-time Transmission}).}
  \label{fig-workflow}
  \vspace{-12pt}
\end{figure}
\section{Introduction}
Local deployment of LLMs is an increasingly popular paradigm, chosen by many organizations and individuals for private or domain-specific application~\cite{singhal2025toward,hau2025llms,On-Premise-adv,kumar2024beyond}. \deleted{However, the rapidly advancing capabilities of LLMs have enabled various AI services while simultaneously giving rise to abuse-related issues.}
\added{However, LLMs incorporate increasingly diverse task capabilities that exceed the usage requirements of specific deployers, thereby increasing abuse-related risks in local deployment scenarios.}
For instance, students can use LLMs to generate assignments or even academic misconduct articles~\cite{ELALI2023100706}, while malicious users can exploit models like GPT to write harmful code and phishing emails~\cite{wiggers2023chatgpt}.

\begin{table*}[tbhp]
\caption{Comparison of Potential Security Mechanisms for Mitigating Abuse of Local LLMs}
\resizebox{\linewidth}{!}{
\begin{tabular}{@{}cccccccccc@{}}
\toprule
\multirow{2}{*}{\textbf{Method}} & \multirow{2}{*}{\textbf{Limit UnAuth. Task}} & \multicolumn{3}{c}{\textbf{Reliable}} & \multicolumn{2}{c}{\textbf{Flexibile}} & \multirow{2}{*}{\textbf{Data Protect}} \\ 
\cmidrule(lr){3-5} \cmidrule(lr){6-7}
                                 &                                              & \textbf{Limit UnAuth. Capability}  & \textbf{Robust$^*$} & \textbf{Intrinsic$^\text{\ding{72}}$} & \textbf{Customized} & \textbf{Dynamic$^+$} &                                        \\ \midrule
Distillation/Fine-tunning~\cite{rozière2024codellamaopenfoundation,isarth_distill_gpt2_story_generator}& \Circle & \Circle & \Circle  &  \CIRCLE   & \LEFTcircle$^{\dagger}$ & \Circle & \Circle \\
Safety Alignment~\cite{ouyang2022traininglanguagemodelsfollow,10.5555/3600270.3602281,10.5555/3294996.3295184,foster2023fastmachineunlearningretraining}& \CIRCLE & \Circle & \Circle & \CIRCLE & \Circle & \Circle & \Circle \\
Unlearning~\cite{pochinkov2024dissectinglanguagemodelsmachine,9519428,foster2023fastmachineunlearningretraining}& \CIRCLE & \CIRCLE & \CIRCLE & \CIRCLE & \LEFTcircle$^{\dagger}$ & \Circle & \CIRCLE \\
Watermarking~\cite{zhang2024watermarking,liang2024watermarking,xu2025mark,ye2025periodic}& \Circle & \Circle & \Circle & \Circle & \Circle & \Circle & \Circle \\
Malicious I/O Detection~\cite{cao-etal-2024-defending,xie-etal-2024-gradsafe,zhang-etal-2024-defending}& \CIRCLE & \Circle & \Circle & \Circle & \CIRCLE & \CIRCLE & \Circle \\ 
\textbf{\textsc{SecNeuron}}& \CIRCLE & \CIRCLE & \CIRCLE & \CIRCLE & \CIRCLE & \CIRCLE & \CIRCLE \\ 
\bottomrule
\end{tabular}
}
\begin{tablenotes}
        \small
        \item 1. \textbf{Limit UnAuth. Task} refers to restricting LLMs from completing unauthorized tasks; \textbf{Limit UnAuth. Capability} directly limits the model's underlying capabilities, making it inherently unable to perform unauthorized tasks (low performance), and is more robust and reliable.
        \item 2. $^*$: robustness against malicious prompt for abuse, like jailbreak or prompt injection. $^\text{\ding{72}}$: is extremely important for local deployment, as temporary security mechanisms can be easily bypassed or removed locally. $^{\dagger}$: requires significant overhead to fine-tune or maintain multiple versions of LLMs, which is impractical. $^+$: tasks of local LLMs can be dynamically adjusted with minimal overhead.

      \end{tablenotes}
\label{tab-compare}
\vspace{-8pt}
\end{table*}

Preventing the abuse of LLM has gradually become a focal point of both societal and academic concern. Existing solutions can primarily be categorized into two aspects: 1) temporary defences that constrain user interactions, such as malicious input/output (I/O) filter~\cite{cao-etal-2024-defending, xie-etal-2024-gradsafe, zhang-etal-2024-defending} or \deleted{access control}\added{safety system prompt} and 2) adjust the LLMs to refuse inappropriate queries or tailor for specific tasks, such as safety alignment~\cite{ouyang2022traininglanguagemodelsfollow,10.5555/3600270.3602281,10.5555/3294996.3295184,foster2023fastmachineunlearningretraining}or task-specific fine-tuning/distillation~\cite{rozière2024codellamaopenfoundation,isarth_distill_gpt2_story_generator,Selvaraju_2018_ECCV,liu2018understanding,song2024does}. These works predominantly address scenarios where the LLM is deployed in the cloud, meaning that the \textbf{runtime environment is controlled by the model developer and the abuse behaviours are pre-known} (e.g., according to human values). 
Consequently, their effectiveness relies on two assumptions: 1) constraints on user interactions should be honestly enforced, and 2) sufficient resources and capabilities to adjust the LLM according to explicitly defined abuse behaviours. These assumptions make it challenging to apply these methods to local deployments.

Local LLMs are deployed on private clouds or local PCs, where \textbf{the runtime environment is controlled by the model deployer (users who might be malicious).} Moreover, deployment requirements may be customized \deleted{(e.g., deployers may only access a subset of the model's capabilities, such as a hospital not requiring the model’s code generation abilities)}and dynamic (e.g., the function of a local LLM in a family may vary depending on the intended user - certain capabilities should be restricted when children access the model). \deleted{Therefore, abuse in local deployments not only involves the pre-known generation of malicious content by LLMs but also includes the unauthorized use of model functionalities by malicious users (deployers).}\added{Therefore, abuse in local deployments extends beyond the pre-known generation of malicious content to include using any unauthorized model functionalities by malicious deployers.} \added{For instance, a hospital might require only the medical diagnosis capabilities of LLMs, while needing to strictly limit its code generation ability.} It is evident that the two aforementioned assumptions are difficult to meet in this context: 1) deployer-controlled environments cannot impose constraints on deployer behaviour, and 2) adjusting the model for various deployment requirements is prohibitively expensive. 
In summary, the expansion of abuse definition and changes in deployment scenarios make existing methods challenging to adapt, as shown in Table~\ref{tab-compare}. Notably, while Trusted Execution Environments (TEEs)~\cite{10.1145/3386901.3388946,PhalaNetwork2024} provide strong security protection for local LLMs, they focus primarily on the confidentiality and integrity of critical parameters during the running time and cannot prevent abuses.
\begin{figure}[t]
  \centering
  \includegraphics[width=\linewidth]{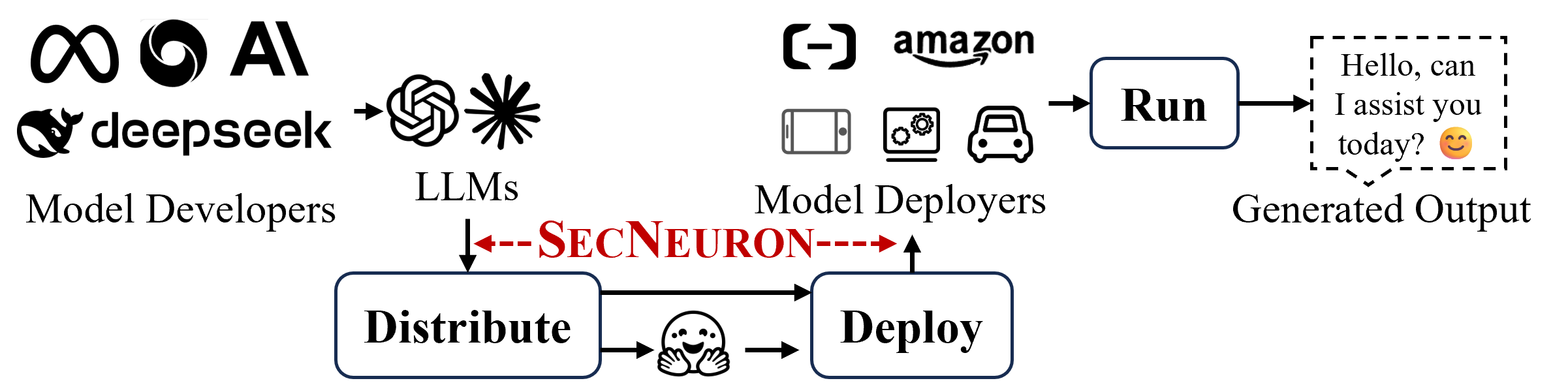}
  \caption{Pipeline for local deployment of LLMs. \textsc{SecNeuron} enhances security during distribution and deployment.}
  \label{fig-scope}
  \vspace{-10pt}
\end{figure}

There is an urgent need for effective methods to dynamically implement task-level abuse control for locally deployed LLMs. As previously discussed, imposing constraints on deployer behaviour in deployer-controlled environments is impossible. Hence, \textbf{we propose that the abuse control mechanisms for local LLMs must directly operate on the instinct capabilities of models, i.e., limiting capabilities on unauthorized tasks.} For example, even if users employ adversarial prompts, an LLM without code generation capabilities would still be unable to generate malicious code (e.g., Figure~\ref{fig-example-code} in \S 6.3). Besides, the mechanisms should be efficient and flexible for LLMs with millions of neurons to support customized deployment.
It brings two challenges that must be addressed:

$\bullet$ \textbf{\textit{C1: How to ensure that capabilities on unauthorized tasks are limited without disrupting the authorized tasks?}}
We aim to limit the LLMs' capabilities (performance on generating next tokens) on unauthorized tasks rather than merely restricting the completion of these behaviors (e.g., safety alignment or I/O detection).
Deactivating or pruning important neurons for specific tasks, like unlearning~\cite{pochinkov2024dissectinglanguagemodelsmachine}, could be an effective method to limit their capabilities.
However, LLMs are optimized for multitasks during training, making it challenging to fully isolate these logically related neurons for different tasks~\cite{raffel2023exploringlimitstransferlearning,caruana1997multitask}, restricting unauthorized tasks may inadvertently affect others. 
Therefore, more effective mechanisms are required for \added{algorithmically} decoupling task-specific neurons and strategies to address unintended coupling.

$\bullet$ \textbf{\textit{C2: How to enable local LLMs adaptable to customized deployer permission with minimal overhead?}}

Existing solutions for customized LLMs (Unlearning or distillation/pruning for specific tasks) are inherently irreversible. Developers require significant resources for fine-tuning or maintaining and transmitting multiple versions of LLMs to meet different user permissions, with the overhead increasing linearly with the number of permissions.
Neurons for specific tasks should be dynamically forcibly disabled or activated (reversible) to meet the customized and dynamic permission requirements of different deployment scenarios.

Fortunately, we found that cryptographic tools can provide a reversible and training-free method tailored for customized capability limitation: encrypting specific neurons limits capability on certain tasks while decrypting them restores the capability. 
Based on this intuition, we designed and implemented \textsc{SecNeuron}, a secure mechanism that enables one-time encryption of the LLM, with different users dynamically decrypting and gaining different capabilities (Figure~\ref{fig-workflow}).
Compared to the methods in Table~\ref{tab-compare}, \textsc{SecNeuron} is characterized by its cost-effectiveness, flexibility, and reliability.

\textbf{\textsc{SecNeuron} DESIGN.} 
The core of \textsc{SecNeuron} are neuron encryption and selective decryption: deployers can dynamically decrypt the neurons they are authorized to access, executing only the authorized tasks.
Firstly, we designed a penalty-based task-specific neuron extraction mechanism to enhance existing neuron importance analysis methods complemented by an efficient mechanism for handling coupled neurons (\textbf{\textit{Addressing C1}}).
Then, we propose a hybrid encryption framework, particularly designed for LLMs with millions of neurons, that balances the flexibility of attribute-based encryption with the efficiency of symmetric encryption (\textbf{\textit{Addressing C2}}).
\textit{Policy Layer}: Neurons are assigned different keys and access policies based on task relevance. The Ciphertext-Policy Attribute-Based (CP-ABE) with a carefully designed policy tree is used to manage keys and coupled neurons across different tasks.
\textit{Execution Layer}: Advanced Encryption Standard (AES) is employed for neuron parameters encryption and decryption; each neuron only needs to be encrypted once.
Deployers dynamically obtain decryption keys based on their attributes, allowing them to decrypt only the authorized portions.
Since undecrypted neurons can degrade the overall performance of partially decrypted LLM, we designed an undecrypted neuron detection mechanism based on the randomness distribution of ciphertext for adaptive pruning. 
\textsc{SecNeuron} offers the following advantages:

\textbf{Flexible and Efficient.} \textsc{SecNeuron} implements an efficient mechanism with one-time encryption \& transmission, and multi-party decryption (Figure~\ref{fig-workflow}), enables dynamic capability updates and flexible permission configuration based on user attributes. For example, permissions such as \textit{(Institution = Hospital) AND (Licence = True)} can restrict access to LLMs' diagnosis functionality.

\textbf{Reliable and Certified.} \textsc{SecNeuron} enforces certified capability constraints through neuron-level encryption. Once a neuron's association with an unauthorized task is explicitly identified, \textsc{SecNeuron} can theoretically ensure that it cannot be activated or utilized, providing a provable safeguard against task abuse.




We summarize three  contributions of this work:

$\bullet$
We propose a novel abuse mitigation mechanism -\textsc{SecNeuron}- which \added{creatively integrates classic access control policies with the intrinsic capabilities of LLMs}\deleted{integrates cryptographic techniques to constrain the intrinsic task-specific capabilities of LLMs}, enabling flexible, reliable, and certified abuse control even under deployer-controlled environments. To the best of our knowledge, \textsc{SecNeuron} is the first dynamic task-level capability management method for LLMs and can serve as a plugin to secure existing deployment pipelines (Figure~\ref{fig-scope}).

$\bullet$ 
We introduce a task-specific reusable neuron decoupling and managing algorithm that enables task-grained capability control at the neuron level. We further propose a hybrid hierarchical encryption framework to support efficient and flexible encryption and decryption of millions of neurons. Additionally, we develop a ciphertext distribution-based neuron identification algorithm to ensure the effectiveness of partially decrypted LLMs.



$\bullet$ 
\textsc{SecNeuron} effectively limits the accuracy of unauthorized tasks to below 25\% while ensuring authorized tasks are impacted by less than 2\%.
It also prevents unauthorized training data extraction, with success rates of PII lower than 5\% and MIA nearly 0\%. 
Furthermore, \textsc{SecNeuron} reduces the encryption and transmission overhead associated with permissions from linear to constant levels, requiring only millisecond-level key generation and byte-level key exchange for local LLM capability updates.
\begin{figure}[t]
  \centering
  \includegraphics[width=\linewidth]{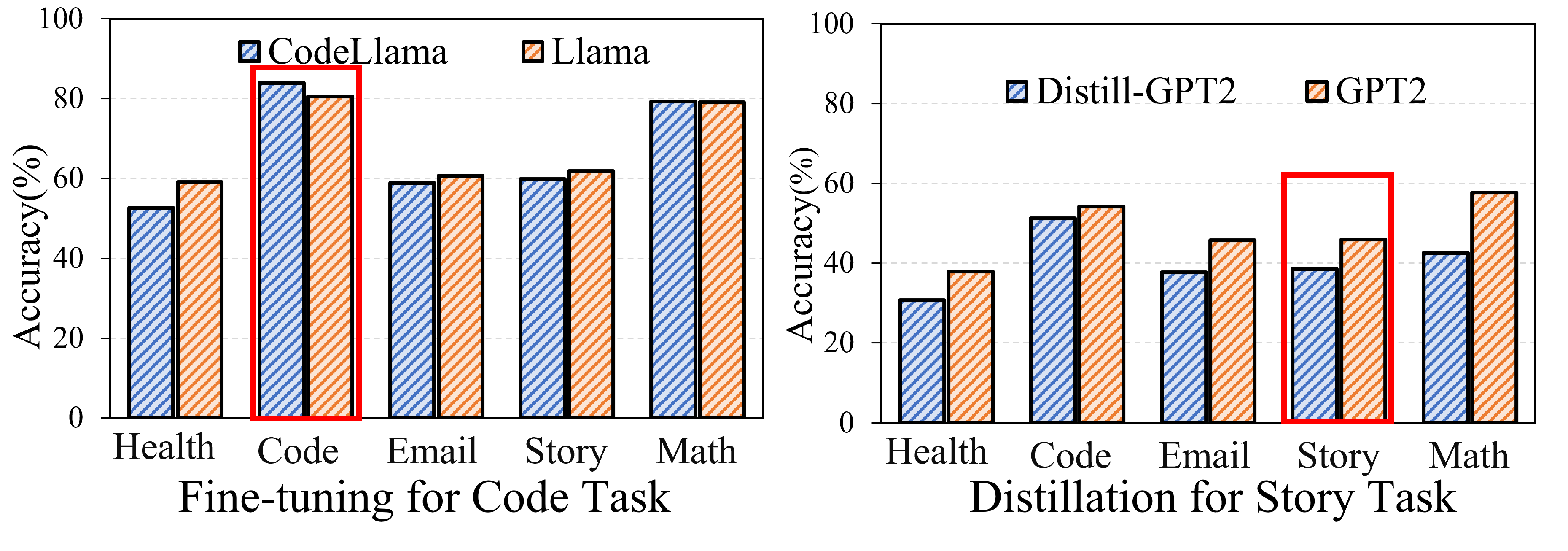}
  \caption{Multi-task performance after fine-tuning for Code task (CodeLlama~\cite{rozière2024codellamaopenfoundation}) or distillation for Story task (Distill-GPT2~\cite{isarth_distill_gpt2_story_generator}). The capabilities of other tasks have not been significantly limited and can still be abused.}
  
  \label{fig-finetune}
\end{figure}
\section{Background and Related Works}
\subsection{Motivating Use Cases}
\added{\textsc{SecNeuron} offers developers a framework for controlled LLM distribution in local deployment scenarios.} This section outlines the motivating use cases for \textsc{SecNeuron}, including existing and future scenarios.

\textbf{Secure and Controllable Model Publishing.}
Plenty of LLMs are published on platforms like Hugging Face, allowing users to download and deploy them locally. \textsc{SecNeuron} can be seamlessly integrated into existing pipelines, providing a low-cost solution to enhance both the security (encrypted model transmission) and controllability during the distribution process (Figure~\ref{fig-scope}). 

\added{\textbf{Fine-grained Commercial Licensing.} Model developers create large-scale systems capable of executing diverse high-value tasks concurrently. SecNeuron facilitates domain-specific licensing to various clients, enabling customizable task activation according to individual client needs.}

\textbf{Dynamic On-device Deployment.} An increasing number of smart devices come pre-installed with LLMs. When device permissions change (e.g., switching to child mode), \textsc{SecNeuron} can activate or disable specific tasks (e.g., social content) through a lightweight key exchange, avoiding re-downloading the LLM.  

\deleted{\textbf{Limiting Malicious Behavior in LLMs.} Harmful content generation by the model (such as discriminatory speech) can be treated as unauthorized tasks. Through \textsc{SecNeuron}, their neurons are forcibly encrypted, and keys for such tasks are not distributed.}

\subsection{Security Issues: Multi-level Abuse}

\textbf{Task Abuse (Model Level).}\deleted{The versatility of LLMs enables them to perform a wide range of tasks (e.g., code generation and medical diagnosis). However,} Once the model is distributed, developers lose control over how the model is used, creating risks of performing unauthorized tasks. 
On the one hand, it poses a serious threat to developers' intellectual property, as the functionality of the LLM represents algorithmic innovations and the significant costs associated with training. On the other hand, this may violate legal boundaries. For example, malicious users may exploit unauthorized \deleted{email generation capabilities to create phishing emails or leverage}coding capabilities to generate harmful scripts~\cite{wiggers2023chatgpt}.

\textbf{Training Data Extraction (Data Level).} LLMs may memorize details from their training data, making it possible for malicious users to reverse-engineer training data from the model's outputs. 
\deleted{For example, PII extraction attacks~\cite{lukas2023analyzing} can retrieve personal identification information contained in the training data, while membership inference attacks~\cite{puerto2024scalingmembershipinferenceattacks} can determine whether specific data was included in the training dataset.}
The training dataset, containing domain-specific knowledge and trade secrets, is also essential for developer's intellectual property. Moreover, some datasets include large amounts of sensitive data (such as PII in Enron Email Dataset~\cite{enron_dataset}), posing significant security risks.

\subsection{Related Works for Mitigating Abuse}

\textbf{Distillation \& Pruning \& Fine-tuning for Specific Task.} 
These approaches are designed to adapt LLMs to specific tasks~\cite{rozière2024codellamaopenfoundation,isarth_distill_gpt2_story_generator,Selvaraju_2018_ECCV,liu2018understanding,song2024does}, which effectively preserve LLMs' capabilities for authorized (target) tasks but do not impose strict constraints on unauthorized ones. As illustrated in Figure~\ref{fig-finetune}, even after fine-tuning or distillation, the model retains its capability to perform other tasks, which may lead to potential abuse.

\textbf{Safety Alignment.} Safety alignment~\cite{ouyang2022traininglanguagemodelsfollow,10.5555/3600270.3602281,10.5555/3294996.3295184,foster2023fastmachineunlearningretraining} aligns LLMs with specific safety objectives and constraints through fine-tuning or reinforcement learning, ensuring their behaviour adheres to authorized tasks. Such methods, which primarily rely on refusal to respond, can limit the behaviour of unauthorized tasks but fail to fundamentally restrict the underlying capabilities of LLMs, leaving them vulnerable to adversarial prompts~\cite{10.5555/3666122.3669630,zou2023universaltransferableadversarialattacks,10.1145/3658644.3690346,10.1145/3658644.3690346}.

\textbf{Unlearning}. Unlearning~\cite{pochinkov2024dissectinglanguagemodelsmachine,9519428,foster2023fastmachineunlearningretraining} aims to remove or mitigate knowledge or patterns that a model has previously learned, aligning closely with our goal of limiting the model’s capabilities. However, unlearning is irreversible, meaning that once capabilities for a specific task of an LLM are restricted, they cannot be restored, making it unsuitable for dynamic, customized local deployment. 


\textbf{Watermarking.}
Watermarking techniques embed invisible markers for copyright verification by modifying model parameters~\cite{szyller2021dawndynamicadversarialwatermarking,rouhani2018deepsigns,chen2019blackmarks} or output distribution~\cite{abdelnabi2021adversarial}. Recently, numerous watermarking methods tailored for LLMs have been proposed~\cite{zhang2024watermarking,liang2024watermarking,xu2025mark,ye2025periodic}, especially,~\cite{xu2025mark} proposes watermark-based protections to address misuse of local LLMs. However, these approaches fall under post hoc detection methods and cannot proactively prevent misuse.

\textbf{Malicious I/O Detection.} They work by monitoring and restricting behavior on unauthorized tasks through external input and output detection~\cite{cao-etal-2024-defending,xie-etal-2024-gradsafe,zhang-etal-2024-defending}. It is widely used in cloud-based LLM applications and is potentially an efficient method for customizing model tasks across different deployment scenarios. However, such temporary solutions can be easily bypassed or removed when LLM is deployed in deployer-controlled environments~\cite{rando2025do}.
\subsection{CP-ABE \& AES-CTR}
This section introduces cryptography algorithms used in the paper.

\textbf{Ciphertext-Policy Attribute-Based Encryption (CP-ABE)}.
CP-ABE~\cite{4223236,cryptoeprint:2019/966} is an advanced cryptographic technique that enables fine-grained access control over encrypted data.
Data owners define access policy, which specifies the conditions under who can decrypt. Decryption is only possible if the user's attributes satisfy the access policy embedded in the ciphertext.
A CP-ABE Cryptor ($E_1$ in Algorithm ~\ref{alg:enc},\ref{alg:dec}) includes four main phases:

(1)\texttt{Setup} ($PK,MSK \leftarrow E_1.setup()$): Initialize bilinear group $G_1$ and target group $G_T$; generate $PK$ and $MSK$.

(2)\texttt{Encrypt} ($C_p\leftarrow E_1.encrypt(PK,M,p)$): use the $PK$ to encrypt $M$ and embed the access policy $p$ into the ciphertext $C_p$. $p$ is often represented as a tree, where the nodes correspond to logical operators such as AND, OR, and threshold gates.

(3)\texttt{KeyGen} ($SK \leftarrow E_1.keyGen(PK,\mathcal{A})$): generate attribute-based secret key $SK$ by $PK$ and attribute $\mathcal{A}$.

(4)\texttt{Decrypt} ($M'\leftarrow E_1.decrypt(PK,SK,C_p)$): attempt to decrypt the ciphertext $C_p$ using $SK$.
If attributes $\mathcal{A}$ satisfy the access policy in the ciphertext, the decryption succeeds; otherwise, it fails.


\textbf{Advanced Encryption Standard - Counter Mode (AES-CTR).}  AES-CTR~\cite{selent2010advanced,nechvatal2001report} is a widely used symmetric encryption mode that transforms AES into a stream cipher. By combining a unique nonce and an incrementing counter, AES-CTR generates a keystream derived from the master key, which is then XORed with the plaintext for encryption or with the ciphertext for decryption. It is highly efficient, parallelizable, and supports random access to encrypted data, making it suitable for applications. 




\begin{figure*}[ht]
  \centering
  \includegraphics[width=\linewidth]{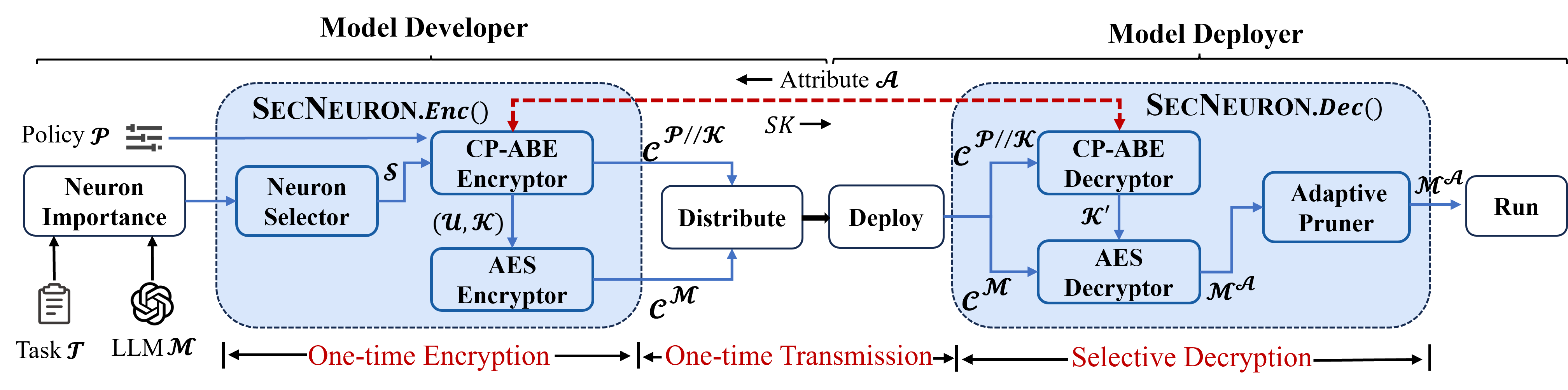}
  \caption{Overview of \textsc{SecNeuron} framework. \textsc{SecNeuron} (blue part) serves as a plug-and-play secure tool in model distribution and local deployment pipeline.
  When permissions change, developers can adjust the capability of the local LLM by only simple key exchanges (red dotted line) with $O(1)$ complexity ($ms-level$ computation and $B-level$ transmission overhead), eliminating re-encrypt and re-transmit LLM with $O(\mathcal{M})$ complexity ($mins-level$ computation and $GB-level$ transmission overhead).
}
  \label{fig-overview}
  \vspace{-10pt}
\end{figure*}

\section{Overview}

\subsection{Threat Model}
As shown in Figure~\ref{fig-workflow}, our system consists of two primary entities: the \textit{\textbf{Model Developer}} and \textit{\textbf{Model Deployer}}.

\textit{\textbf{Model Developer}}.
\deleted{Model developers train LLMs using proprietary or high-value datasets, enabling these models to perform a wide range of downstream tasks.} 
\deleted{Once trained, the models are distributed to model deployers for local deployment and use, completely beyond the developer's control.} 
\deleted{Model developers can define task-level access policies for their LLMs tailored to different deployers, aiming to ensure the proper uses of their LLMs from both Model Level and Data Level for each deployer (see §2.2).} 
\added{
Model developers train LLMs with the capability to perform a wide range of tasks, and then distribute these well-trained models to model deployers for local deployment and use. The training process may involve high-value or diverse datasets, as well as proprietary architectures and training methodologies. Consequently, developers may seek to distribute these models under controlled conditions for local deployment, specifically by defining task-level access policies tailored to different deployers to ensure the proper use of their LLMs.
}



\textit{\textbf{Model Deployer}}.
The model deployer, in this context, acts as the adversary. Different deployers have dynamic and customized deployment requirements, often a specific task subset of the LLM.
\textbf{They have access to a locally deployed (fully white-box) LLM with authorization for some tasks (Partially Authorized Users) or none at all (Passive Attacker).} Any capabilities of LLMs that exceed the authorized scope of deployers pose the risk of abuse through the following malicious behaviours:

$\bullet$ \textit{Unauthorized Task Abuse} aims to illegally invoke the LLM to perform unauthorized tasks.

$\bullet$ \textit{Training Data Extraction} represents a more advanced attacker aiming to infer unauthorized training data from the model outputs. We specifically consider two common training data extraction attacks for LLMs: membership inference~\cite{puerto2024scalingmembershipinferenceattacks} and PII extraction~\cite{lukas2023analyzing}.

\textbf{Usage Scope.} \textsc{SecNeuron} \added{offers a secure framework for controlled model distribution and local deployment that mitigates potential abuse of locally deployed LLMs.}
\deleted{provides a secure mechanism during the distribution and local deployment phases for mitigating abuse of local LLMs.}
Ensuring security during runtime from parameter theft has been extensively studied~\cite{10.1145/3386901.3388946,PhalaNetwork2024}. \textsc{SecNeuron} is orthogonal to them and can combine for more comprehensive protection (\S 8).
\textsc{Secneuron} aims to prevent \deleted{unauthorized abuse of LLM capabilities}\added{users from abusing  LLM capabilities for which they lack authorization} and can resist collusion, \added{i.e., performing unauthorized tasks by combining partially authorized keys or permissions.}. 
\added{The reconstruction of full model capabilities by collusive users who possess authorization for all tasks falls outside the scope of this study, as such attacks stem from underlying access policy vulnerabilities (Details in \S 5.3).}
\deleted{However, collusion threats arising from access policy vulnerabilities are beyond the scope (Details in \S 5.3).}

\subsection{Design Goals and Formulation}
LLMs are fine-tuned using specific training datasets to enable specific downstream tasks. Based on this, we divide tasks $\mathcal{T}$ according to the training datasets $\mathcal{D}$, meaning that different datasets $D_t$ correspond to different tasks $t$. 
The protection scope of \textsc{SecNeuron} can be summarized as a tuple $< \mathcal{D},\mathcal{T}>$, i.e., the high-value datasets and the model capabilities trained on them, restricting them from being unauthorized abuse and extraction.

Given a well-trained LLM $\mathcal{M}$ and a set of tasks $\mathcal{T}$, \textsc{SecNeuron} achieves secure and controllable distribution and deployment through encryption and selective decryption of neurons in LLM:

Developers define access policy $\mathcal{P}$ and encrypt LLM using \textsc{SecNeuron} (Algorithm~\ref{alg:enc}: $\mathcal{C}^\mathcal{M},\mathcal{C}^\mathcal{P//K} \leftarrow \textsc{SecNeuron}.Enc(\mathcal{M},\mathcal{T},\mathcal{P})$); 

Deployers decrypt the encrypted model based on their attribute $\mathcal{A}$ (Algorithm~\ref{alg:dec}: $\mathcal{M}^\mathcal{A} \leftarrow \textsc{SecNeuron}.Dec(\mathcal{C}^\mathcal{M},\mathcal{C}^\mathcal{P//K},\mathcal{A})$).

$\mathcal{M^A}$ needs to meet the following objectives:

\textit{1) Preserve the performance of authorized tasks.} 
For any task $t$ in the authorized task set $\mathcal{T_A}$, the performance of $\mathcal{M^A}$ is essentially consistent with that of $\mathcal{M}$.
\begin{equation}
    \forall t \in \mathcal{T_A}, \mathcal{P}_t(\mathcal{M}) - \mathcal{P}_t(\mathcal{M^A})< \gamma_t, 
\end{equation}
$\gamma_t$ is a small constant, where smaller indicate greater similarity in performance. $\mathcal{P}_t(\mathcal{M})$ represents the performance (generating correct next tokens) on task $t$ for $\mathcal{M}$, which, in this paper, is measured using $Accuracy$ (Eqution \ref{eq:accuracy}).

\textit{2) Limit the capabilities of unauthorized tasks.}
For task $t$ in the unauthorized task set $\mathcal{T_U}$, performance of $\mathcal{M^A}$ should be limited.
\begin{equation}
    \forall t \in \mathcal{T_U}, \mathcal{P}_t(\mathcal{M^A})< \delta_t,
    \label{eq:limit}
\end{equation}
$\delta_t$ represents a small positive number that defines the upper bound of the model's performance on unauthorized tasks.

\textit{3) Prevent extraction of unauthorized training data.}
For task $t$ in unauthorized task set $\mathcal{T_U}$, prevent the extraction of $D_t$ from $\mathcal{M^A}$:
\begin{equation}
   \forall t \in \mathcal{T_U}, MIA(\mathcal{M^A},D_t)<\epsilon, PII(\mathcal{M^A},D_t)<\epsilon, 
\end{equation}

$MIA(\cdot)$ (Membership Inference Attack)~\cite{puerto2024scalingmembershipinferenceattacks} and $PII(\cdot)$ (Personally Identifiable Information Attack)~\cite{lukas2023analyzing} refer to the attack success rate of two commonly used training data extraction methods. $\epsilon$ is a small positive number that limits the success rate.

Developers can adjust $\gamma_t,\delta_t,\epsilon$ according to the value or security requirements of the task $t$ (data). For example, for highly sensitive and valuable \textit{Health} tasks, a smaller $\delta_t$ can be set, whereas, for $Story$ task, the restrictions on $\delta_t$ can be relaxed.

\textsc{SecNeuron} must also fulfill the following practicality objectives:

\textit{1) Low Overhead.}
Encryption and decryption must have efficient computational overhead for plenty of neurons, and the encrypted model should not introduce excessive transmission overhead.

\textit{2) Flexibility.}
\textsc{SecNeuron} should support complex permission configurations, allowing locally downloaded LLMs to achieve efficient capability adjustments based on the permission of deployers.

\subsection{Challenges and Solutions}

$\bullet$ \textbf{\textit{C1: How to ensure that capabilities on unauthorized tasks are limited without disrupting the authorized tasks?}}

Deactivating or pruning neurons corresponding to unauthorized tasks seems like an effective way to limit capability on these tasks. However, modern LLMs are inherently multitask systems designed to handle a wide range of tasks by leveraging shared neuron representations, meaning multiple tasks may activate the same neurons.
Existing pruning methods~\cite{Selvaraju_2018_ECCV,liu2018understanding,song2024does} overlook the coupling of neurons across different tasks. Specifically, important neurons for unauthorized tasks may also partially contribute to authorized tasks. Naively pruning them could unintentionally degrade the performance of authorized tasks. Similarly, focusing solely on preserving important neurons for the target task does not guarantee restriction of unauthorized tasks (Figure ~\ref{fig-finetune}).

\textbf{Solution.} 
\added{Recognizing that perfect, mutually exclusive neuron isolation for each task is often impractical, our approach focuses on managing and mitigating the effects of neuron coupling.}
To achieve this, we first introduce the penalty factor $\lambda$ to enhance the decoupling of task-specific neurons~\cite{pochinkov2024dissectinglanguagemodelsmachine} (critical for target tasks but insignificant for others).
However, overlapping neurons across different tasks are sometimes unavoidable. Thus, we embed the control of overlapping neurons into the access tree of CP-ABE (Figure~\ref{fig-ac-tree}), eliminating the need for additional management.

$\bullet$ \textbf{\textit{C2: How to enable local LLMs adaptable to customized deployer permission with minimal overhead?}}

Deployer permissions are dynamic and require flexible capability adjustment for local LLMs. 
For example, when children use LLMs, it is necessary to limit the capability of social content generation.
Existing solutions for customized LLMs are irreversible, imposing prohibitive overhead on developers and deployers: Developers must maintain multiple task-specific model versions (including encryption, storage, and transmission), and deployers must repeatedly download and deploy the corresponding LLM versions.
As the complexity of permission combinations increases, the management cost and overhead grow linearly.
While task-specific distillation or pruning~\cite{10.5555/3495724.3497435,10.5555/3454287.3455544,10.5555/3524938.3525491} can reduce the cost of a single transmission, the cumulative cost of multiple transmissions remains significant.
\textbf{Solution.}
From the cryptographic perspective, we propose a reversible capability limitation mechanism for customized LLM needs: encrypting neurons limits capability on certain tasks while decrypting them restores the capability.
Specifically, we introduce a hybrid encryption mechanism to handle the vast number of parameters in LLMs, balancing the flexibility of CP-ABE and the efficiency of AES:

\textit{1)Policy Layer} (\textbf{CP-ABE Encryptor}): Assigns AES keys to neurons based on their task relevance and binds access policies to these keys. Then, CP-ABE is used to encrypt and manage them.

\textit{2)Execution Layer} (\textbf{AES Encryptor}): Uses AES-CTR to encrypt neuron parameters in LLMs, with AES keys for each neuron generated and managed by \textit{Policy Layer}.

Deployers download the complete encrypted model and then obtain the authorized keys based on their attributes (permissions) to decrypt and utilize permitted tasks adaptively.

\begin{algorithm}[t]
\LinesNumbered
\KwData{Original Model: $\mathcal{M}$, Task List: $\mathcal{T}$, Policy $\mathcal{P}$ }
\KwResult{Encrypted Model: $\mathcal{C}^\mathcal{M}$,Policy $\&$ Key Encrypted by CP-ABE: $\mathcal{C}^\mathcal{P//K}$}
$\mathcal{C}^\mathcal{P//K} \leftarrow \emptyset, \mathcal{C}^\mathcal{M}\leftarrow \mathcal{M}$\;
Create CP-ABE Cryptor $E_1$ and AES Cryptor $E_2$\;
$\blacktriangleright$ \textbf{Neuron Selector} \\
Select important neurons for each task: $\mathcal{S} \leftarrow \textsc{SecNeuron}.select(\mathcal{M},\mathcal{T})$\;
$\blacktriangleright$ \textbf{CP-ABE Encryptor:} $\boldsymbol{Enc_{ABE}(PK,\mathcal{P}//\mathcal{K})\rightarrow \mathcal{C}^\mathcal{P//K} }$\\
CP-ABE Init.: $PK,MSK \leftarrow E_1.setup()$\;
Decompose all possible combinations of $\mathcal{S}$ into disjoint subsets: $\mathcal{U} \leftarrow \{\bigcap_{t \in \mathcal{T'}} \mathcal{S}_t \setminus \bigcup_{t \in \mathcal{T}\setminus \mathcal{T'}} \mathcal{S}_t : \mathcal{T'} \subseteq \mathcal{T}, \mathcal{T'} \neq \emptyset\} \cup \{\Omega \setminus \bigcup_{t \in \mathcal{T}} \mathcal{S}_t\}$  \;
\ForEach{subset $\mathcal{N} \in \mathcal{U}$}{
Select a random element $k_{gt}$ from the target group GT\;
Create access policy $p$ for $k_{gt}$ based on $\mathcal{P}$\;
Encrypt $k_{gt}$ and $p$ using $PK$ by $E_1$: $c\leftarrow E_1.encrypt(PK,k_{gt},p)$\;
$\mathcal{C}^\mathcal{P//K} \leftarrow \mathcal{C}^\mathcal{P//K} \cup c $\;

$\blacktriangleright$\textbf{AES Encryptor:} $\boldsymbol{Enc_{AES}(\mathcal{K},\mathcal{M})\rightarrow \mathcal{C}^\mathcal{M} }$\\

Generate 64-bit AES key $k$ from $k_{gt}$\;
\ForEach{neuron $n \in \mathcal{N}$}{
    Encrypt $\mathcal{M}$ using $k$ by $E_2$: $\mathcal{C}^\mathcal{M}_n\leftarrow \textsc{SecNeuron}.encrypt(k,\mathcal{M},n,E_2)$ 
}
}

\Return $\mathcal{C}^\mathcal{M},\mathcal{C}^\mathcal{P//K}$\;
\caption{Encryptor: $\textsc{SecNeuron}.Enc()$}
\label{alg:enc}
\end{algorithm}

\section{\textsc{SecNeuron} Design}
As shown in Figure~\ref{fig-overview}, \textsc{SecNeuron} has two components: the Encryptor (Algorithm~\ref{alg:enc}: \textsc{SecNeuron}.$Enc()$) for the model developer and the Decryptor (Algorithm~\ref{alg:dec}: \textsc{SecNeuron}.$Dec()$) for the model deployer. The Encryptor is executed only once to generate an encrypted version of the LLM. Different deployers can adaptively use the Decryptor to access the authorized portions of the encrypted model based on their permissions.
\subsection{Encryptor for Model Developer}

\textbf{Neuron Selector.}
\begin{algorithm}[h]
\LinesNumbered
\KwData{Original Model: $\mathcal{M}$, Task List: $\mathcal{T}$, Importance Score Function $S$, Importance Threshold: $\tau$}
\KwResult{Selected Neurons: $\mathcal{S}$}
$\mathcal{S} \leftarrow \emptyset$\;
\ForEach{task $t\in \mathcal{T}$}{
    $\mathcal{S}_t \leftarrow \emptyset$\;  
    Calculate the task-specific score for each neuron $n$: $s_n\leftarrow S(\mathcal{T},t,n)$\;
    Normalize and Sort $s$ in descending order\;
    Select Neurons: \\ $\mathcal{S}_t \leftarrow \mathcal{S}_t.append(n)$ if $\mathcal{S}_t.sum()+s_n<\tau \cdot s.sum()$ \;
    $\mathcal{S}\leftarrow \mathcal{S}\cup \mathcal{S}_t$\;
}
\Return $\mathcal{S}$
\caption{$\textsc{SecNeuron}.select()$}
\label{algm:selector}
\end{algorithm}
The Neuron Selector is used to identify task-specific neurons (important only for the target task but not for other tasks) for each task in $\mathcal{T}$ (Lines 3-4 in Algorithm~\ref{alg:enc}). Referenced from~\cite{pochinkov2024dissectinglanguagemodelsmachine}, we use the mean of absolute activation to calculate the importance of each neuron, and then we introduce $\lambda$ as a penalty factor to calculate the task-specific score. Let $n$ be a neuron and denote its activations by $z_n$. Given a task $t \in \mathcal{T}$ and its sampled training dataset $D_t$, we define task-specific scoring function $S$ as:
\begin{equation}
    I(t,n):=\frac{\sum_{d\in D_t}z_n(d)}{|D_t|},
\end{equation}
\begin{equation}
    S(\mathcal{T},t,n):=I(t,n)- \lambda \cdot \max_{t'\in \mathcal{T},t' \neq t} I(t',n)
\end{equation}
The larger the value of $S$, the more important neuron $n$ is specific for task $t$. Therefore, we select neurons from the largest $S$ value for each task, continuing until the cumulative sum exceeds the threshold $\tau$.
\textbf{The calculation of $I(t,n)$ can adopt other effective neuron importance estimation methods.}
In particular, the neuron selection algorithm is shown in Algorithm \ref{algm:selector}.

\textbf{CP-ABE Encryptor} \textit{(Policy Layer).}
CP-ABE Encryptor is responsible for key management and does not directly participate in model encryption (Lines 5-12 in Algorithm~\ref{alg:enc}). The process involves three key steps: \textit{CP-ABE Init.}, \textit{AES Key \& Policy Gen.} and \textit{CP-ABE Enc.}:

\textit{1) CP-ABE Init.} Directly invoke the setup mechanism of CP-ABE to generate a public key $PK$ and a master secret key $MSK$, which are used to derive attribute-based secret keys $SK$ for users.

\textit{2) AES Key \& Policy Gen.}
For selected important neurons $S$, we decompose them into multiple disjoint subsets to address the overlap between task-specific neurons of different tasks (the overlapping neurons are treated as a separate subset. For example, neurons from $t1$, neurons from $t2$, and neurons shared by both $t1$ and $t2$ are decomposed into distinct subsets $\mathcal{U}$ = \{$\mathcal{S}_{t1},\mathcal{S}_{t2},\mathcal{S}_{t1}\cap \mathcal{S}_{t2}$\}) (Line 7 in Algorithm \ref{alg:enc}). 
For each subset $\mathcal{N}\in \mathcal{U}$, we randomly select an $k_{gt}$ from $G_T$ group of CP-ABE to generate the AES key $k$, which serves as the encryption key for all neurons in that subset (Figure~\ref{fig-key-ass}).
At the same time, the access policy tree is constructed, assigning policy $p$ to each key (Line 8-10 in Algorithm \ref{alg:enc}).
\added{Furthermore, neurons that are not selected by any task are designated as a separate common subset. These common neurons are also assigned an encryption keys with the access policy that permits any authorized task to access them, thereby enhancing security against passive attackers.}

Access to overlapping neurons is also integrated as part of the policy tree, eliminating the need for additional management steps. Specifically, we divide the policy tree into two layers based on tasks (Figure~\ref{fig-ac-tree}): Neuron-level Policy and User-level Policy. 
To minimize the impact on authorized tasks, Neuron-level Policy employs \textit{OR} nodes to manage the overlap of neurons and is a built-in, immutable mechanism. 
User-level Policy specifies the access rights of deployers with different attributes for each task, allowing developers to adjust the policy flexibly according to specific requirements.

\textit{3) CP-ABE Enc.}
Encrypt all generate $k_{gt}$ and embed the policy $p$ into the ciphertext, with directly using the CP-ABE encryption function $E_1.encrypt(\cdot)$.
\begin{figure}[t]
  \centering
  \includegraphics[width=.97\linewidth]{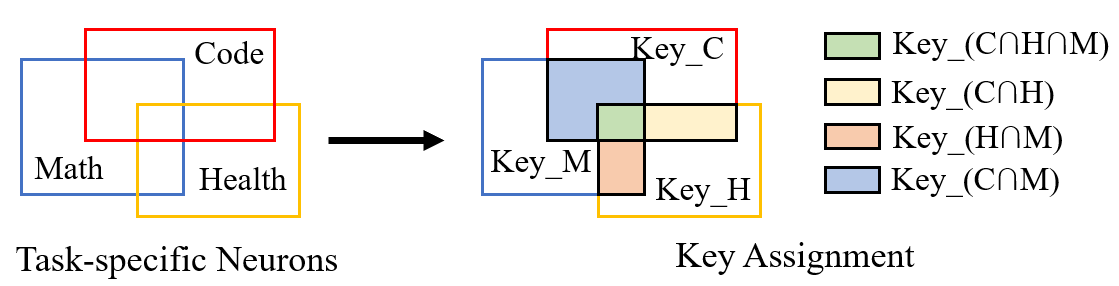}
  \caption{Illustration of AES key assignment.}
  \label{fig-key-ass}
  \vspace{-4pt}
\end{figure}
\begin{figure}[t]
  \centering
  \includegraphics[width=.95\linewidth]{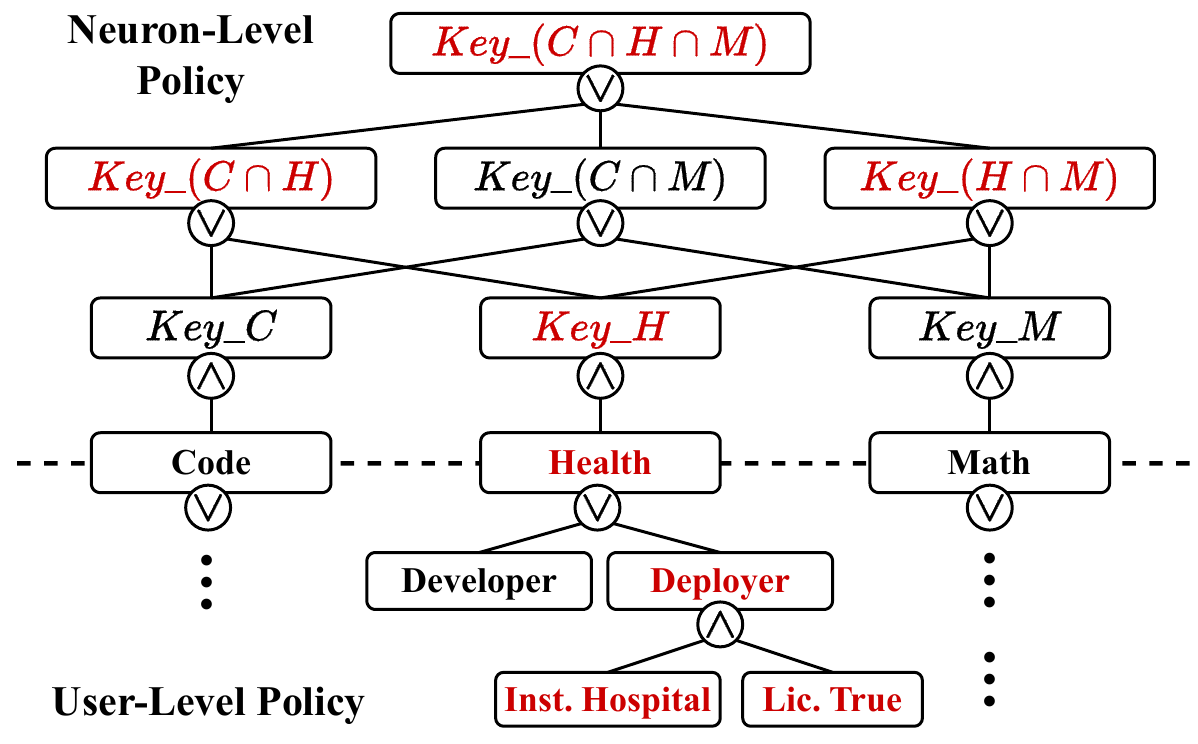}
  \caption{Access policy tree for the CP-ABE encryptor,$\land$ represents logical AND, while $\lor$ represents logical OR. An example of authorized access is highlighted in red: Deployers meet $Institution=Hospital$ $and$ $Licence=True$ are allowed to use the Health task of the LLM. They can access keys such as \{$Key\_H$,$Key\_(H \cap C)$,$...$\} to decrypt the corresponding neurons.}
  \label{fig-ac-tree}
\end{figure}

\textbf{AES Encryptor} \textit{(Execution Layer).}
The AES Encryptor is used for encrypting each neuron in the LLM model by AES-CTR (Lines 13-16 in Algorithm \ref{alg:enc}).
\begin{algorithm}[h!]
\LinesNumbered
\KwData{Encrypt Key: $k$, Original Model: $\mathcal{M}$, Neuron Index: $n$, AES Cryptor $E2$}
\KwResult{Encrypted Neuron parameters: $\mathcal{C}^\mathcal{M}_n$}
$\mathcal{C}^\mathcal{M}_n \leftarrow \mathcal{M}_n$\;
Init CTR Mode: $E2.MODE \leftarrow CTR,$ $E2.COUNTER\leftarrow n$\;
Encrypt $W_{IN}$: $\mathcal{C}^\mathcal{M}_n.W^{IN}_n=E2.encrypt(k,\mathcal{M}_n.W_{IN})$\;
Encrypt $W_{OUT}$: $\mathcal{C}^\mathcal{M}_n.W^{OUT}_n=E2.encrypt(k,\mathcal{M}_n.W_{OUT})$\;
Encrypt $B_{IN}$: $\mathcal{C}^\mathcal{M}_n.B^{IN}_n=E2.encrypt(k,\mathcal{M}_n.B_{IN})$\;
\Return $\mathcal{C}^\mathcal{M}_n$
\caption{$\textsc{SecNeuron}.encrypt()$}
\label{algm:enc-neuron}
\end{algorithm}

Given a neuron $n$ and its input $x_n$, the neuron feedforward process (MLP layer) can be summarized as:
\begin{equation}
    x_n' = \mathcal{M}_n.W_{OUT}\cdot \sigma(\mathcal{M}_n.W_{IN}\cdot x_n+\mathcal{M}_n.B_{IN}),
\end{equation}
where $W_{IN}$ and $B_{IN}$ are the input weight and bias, respectively, $\sigma(\cdot)$is the activation function (e.g., ReLU or Sigmoid), $W_{OUT}$ represents the output weight. As shown in Algorithm \ref{algm:enc-neuron}, AES Encryptor simultaneously encrypts the parameters $\mathcal{M}_n.W_{OUT}$, $\mathcal{M}_n.W_{IN}$ and $\mathcal{M}_n.B_{IN}$ to encrypt a single neuron $n$ (with the same encrypt key).
Moreover, to ensure that the decryption side can randomly decrypt any neuron, we utilize the AES encryption in CTR mode ($E_2.encrypt(\cdot)$) and pass $n$ as the counter value.
\subsection{Decryptor for Model Deployer}
\begin{algorithm}[t]
\LinesNumbered
\KwData{Encrypted Model: $\mathcal{C}^\mathcal{M}$, Policy $\&$ Key Encrypted by CP-ABE: $\mathcal{C}^\mathcal{P//K}$, Attribute List: $\mathcal{A}$, Public key: $PK$}
\KwResult{Decrypted Model: $\mathcal{M^\mathcal{A}}$}
Create CP-ABE Cryptor $E_1$ and AES Cryptor $E_2$\;

$\blacktriangleright$ \textbf{CP-ABE Decryptor:} $\boldsymbol{Dec_{ABE}(PK,\mathcal{C}^\mathcal{P//K},\mathcal{A})\rightarrow \mathcal{K'}}$\\

Request attribute-based secret key $SK$ using $PK$ and $\mathcal{A}$\;
Init authorized keys: $\mathcal{K'}\leftarrow \emptyset$\;
\ForEach{ciphertext $c \in \mathcal{C}^\mathcal{P//K}$}{
    Decrypt $c$ using $PK$ and $SK$ by $E_1$: $k_{gt}\leftarrow E_1.decrypt(PK,SK,c)$\;
    \If{$k_{gt}$ is correct decrypted}{
        Generate 64-bit AES key $k$ from $k_{gt}$\;
        $\mathcal{K'}\leftarrow \mathcal{K'}\cup k$\;
    }
}
$\blacktriangleright$ \textbf{AES Decryptor:} $\boldsymbol{Dec_{AES}(\mathcal{K'},\mathcal{C}^\mathcal{M})\rightarrow \mathcal{M}^\mathcal{A}}$\\

$\mathcal{M^\mathcal{A}}\leftarrow 0$\;
\ForEach{neuron $n$ of $\mathcal{\mathcal{C}^\mathcal{M}}$}{
    \ForEach{key $k \in \mathcal{K'}$}{
    Decrypt $\mathcal{C}^\mathcal{M}_n$ using $k$ by $E_2$: $m \leftarrow \textsc{SecNeuron}.decrypt(k,\mathcal{C}^\mathcal{M},n,E_2)$\;
    \If{$m$ is correct decrypted}{
        $\mathcal{M}^\mathcal{A}_n\leftarrow m$\;
        \textbf{break}\;
    }
    }
    $\blacktriangleright$ \textbf{Adaptive Pruner} \\
    \If{$\mathcal{M}^\mathcal{A}_n == 0$}{
    Prune neuron $n$ from $\mathcal{M^\mathcal{A}}$\;
    }
    
}
\Return $\mathcal{M^\mathcal{A}}$\;
\caption{Decryptor: $\textsc{SecNeuron}.Dec()$}
\label{alg:dec}
\end{algorithm}
\textbf{CP-ABE Decryptor} \textit{(Policy Layer)}. 
The CP-ABE Decryptor derives the authorized AES decryption key based on attributes of deployers (lines 2-9 in Algorithm \ref{alg:dec}). 
Firstly, the deployer requests an attribute-based secret key $SK$ based on their attributes $\mathcal{A}$. Then, $SK$ is used to decrypt $\mathcal{C}^\mathcal{P//K}$ to obtain the authorized $k_{gt}$ (the access policy is already embedded in the ciphertext, so the CP-ABE decryption function $E1.decrypt(\cdot)$ can be directly invoked without explicit permission checks). Finally, \textsc{SecNeuron} convert each correctly decrypted $k_{gt}$ to AES key $k$ for subsequent AES decryption.

\textbf{AES Decryptor} \textit{(Policy Layer)}. AES Decryptor is used to decrypt model parameters. 
We provide two decryption mechanisms: \textit{Transmission-efficient decryption (T-E dec.)} and \textit{Computation-efficient decryption (C-E dec.)}.
Algorithm~\ref{alg:dec} uses \textit{T-E dec.} as an example.

$\bullet$ \textit{Transmission-efficient decryption (T-E dec.):} Transmit only the encrypted LLM (identical size to original LLM) and CP-ABE ciphertext once, eliminating additional transmission overhead but requiring decryption of the entire LLM.

Since there is no metadata assistance, AES Decryptor attempts to decrypt each neuron using all authorized keys (Lines 10-17 in Algorithm \ref{alg:dec}), requiring verification of whether neurons have been correctly decrypted (Line 15 in Algorithm \ref{alg:dec}). AES encryption operates at the byte stream level, causing neurons of different data types to exhibit distinctive characteristics after encryption due to their varied byte-level serialization patterns in memory (Figure~\ref{fig-save}).
\begin{figure}
  \centering
  \includegraphics[width=\linewidth]{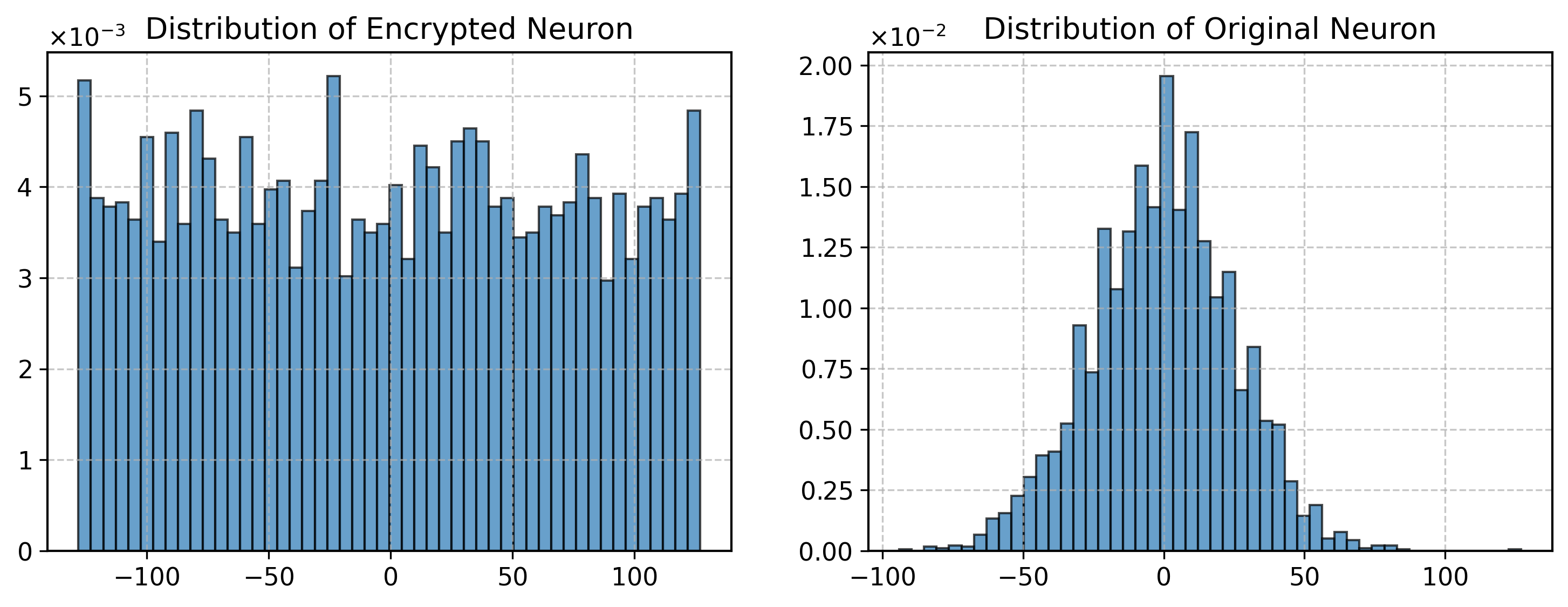}
  \caption{The distribution ($W_{IN}$) of encrypted neurons exhibits notable differences compared to original neurons. }
  \label{fig-distribution}
\end{figure}
Consequently, we propose an efficient undecrypted neuron detection mechanism for two predominant parameter types in LLM: \textit{INT} and \textit{FLOAT}.

\textit{FLOAT16.} Encryption of FLOAT16 parameters may affect the 'exponent' bits, leading to extremely large outlier values $(2^{16})$. Such anomalous values are absent in well-trained LLMs. Therefore, by detecting these outliers in the neuron parameters, we can determine whether it has been decrypted (the same with \textit{FLOAT32}).
    
\textit{INT8.} The range of \textit{INT8} model parameters lies between [-128, 127] and remains invariant even after encryption, precluding the use of outliers for detection. After encryption, the data follows a uniform distribution, whereas the parameters of a trained model exhibit a specific, non-uniform distribution. Therefore, we can determine whether a neuron has been decrypted by its distribution pattern (Figure \ref{fig-distribution}). 

Specifically, for \textit{FLOAT-type} neurons, we select the maximum value of their input matrix $\mathcal{M}_n.W_{IN}$ as the metric ($m$); for \textit{INT-type} neurons, we use the variance of the statistical histogram of $\mathcal{M}_n.W_{IN}$ as the metric ($v_H$). The final determination of whether a neuron is correctly decrypted is made through threshold comparison. We tested the detection effectiveness across different models and achieved a success rate of 100\% in all cases (\S 6.5).

$\bullet$ \textit{Computation-efficient decryption (C-E dec.):} Transmit additional metadata indicating which key each neuron uses, only decrypting the authorized neurons.

When deploying, users download the LLM along with metadata $F$ that indicates the AES key used for each neuron. The size of $F$ equals the number of neurons with transmission overhead smaller than the LLM itself. During decryption, each neuron first retrieves its key based on $F$ and then determines whether it is authorized (whether CP-ABE can decrypt it). If authorized, the neuron is decrypted; otherwise, it undergoes adaptive pruning. When permissions change, only neurons with changed permissions need to be adjusted. For single-task decryption, this approach can reduce decryption overhead by 40\%. Table~\ref{tab:o} analyzes the computational and transmission complexities of different decrypt methods.

\begin{table}[]
\caption{Computational \& Transmission Complexity. }
\resizebox{\linewidth}{!}{
\begin{threeparttable}

\begin{tabular}{@{}ccccc@{}}

\toprule
\multirow{2}{*}{} & \multicolumn{2}{c}{First Deployment} & \multicolumn{2}{c}{Capability Update} \\ \cmidrule(l){2-5} 
 & Computational& Transmission& Computational & Transmission\\ \midrule
Encryptor & $O(\mathcal{M})$  & $O(\mathcal{M})$ & $O(1)$ & $O(1)$  \\
T-E Dec. & $O(|\mathcal{K}'|\cdot \mathcal{M})^*$ & $O(\mathcal{M})$ & $O(|\mathcal{K}'|\cdot \mathcal{M})^*$ & $O(1)$   \\
C-E Dec.$^\dagger$ & $O(\mathcal{M}_t)$ & $O(\mathcal{M}+\mathcal{N}_{neuron})$ & $O(\mathcal{M}_t)$  & $O(1)$ \\ \bottomrule
\end{tabular}

        \begin{tablenotes}
            \small
            \item 1: CP-ABE is used for encrypting and managing keys, with overhead significantly smaller than the processing of LLMs. We disregard its overhead to simplify the complexity analysis.
            \item 2: $^*$ represents the worst-case complexity, $|\mathcal{K}'|$ refers to number of authorized keys.
            \item 3 $^\dagger$: $\mathcal{N}_{neuron}$ refers to the number of neurons, $\mathcal{N}_{neuron}<< \mathcal{M}$ (for a 6.7B LLM, the size of $\mathcal{N}_{neuron}$ is only around 500KB.); $\mathcal{M}_{t}$ refers to authorized portion of LLM for task $t$.
        \end{tablenotes}
\end{threeparttable}
}
\label{tab:o}
\end{table}


\textbf{Adaptive Pruner.} The Adaptive Pruner prunes all unauthorized neurons (lines 18-20 in Algorithm \ref{alg:dec}), accelerating the inference of locally deployed LLMs without affecting the performance on authorized tasks.

\section{Security Analysis}
 In this section, we analyzed the security of \textsc{SecNeuron}.  First, we defined the Task Controllability Principle to ensure that all tasks can be effectively protected and proved the Task Capacity Upper Bound for a given LLM, a necessary but insufficient condition for judging whether a task configuration is reasonable (\S 5.1). Then, we proved that \textsc{SecNeuron} satisfies IND-CPA Security (\S 5.2) and Collusion Resistance Security (\S 5.3).

 
 \subsection{\deleted{Neural Isolation Principle}\added{Task Controllability Principle}}
 When serving as a separate unauthorised task, any task in $\mathcal{T}$ needs to satisfy Equation~\ref{eq:limit}, meaning each task requires a sufficient number of task-specific neurons with no intersection with other tasks. To achieve this goal, we define the Task Controllability Principle.

\textsc{Definition 5.1.}\textbf{(Task Controllability Principle)}
\textit{ Given a LLM $\mathcal{M}$ and tasks $\mathcal{T}$, for each task $t \in \mathcal{T}$, there exists a neuron set $S'_t \subseteq \mathcal{S}_t$ satisfying the following conditions:}

\textit{1) For any two different tasks $t_1, t_2 \in \mathcal{T}$ satisfying: $S'_{t_1} \cap S'_{t_2}=\emptyset$;}

\textit{2) Removing $S'_t$ causes the performance of task $t$ to fall below the target threshold: $\mathcal{P}_t(\mathcal{M} \setminus S'_t)<\delta_t,\mathcal{P}_t(\cdot)$ refers to performance on $t$.}

If satisfying the Task Controllability Principle, the task set $\mathcal{T}$ is controllable, meaning any task $t \in \mathcal{T}$ can be constrained to perform within its target threshold $\delta_t$. \added{Notably, the Task Controllability Principle should not be interpreted as an assumption of absolute physical separation for all neurons involved in a task, but rather as a controllability principle that defines how specific task capabilities can be effectively constrained. 
\textsc{SecNeuron} does not assume a priori that all neurons across different tasks are isolatable, but rather seeks to construct such effective isolation neuron sets.
Despite general neuron entanglement in large LLMs, their large-scale neural architectures allow us to readily pinpoint neuron subsets for each task(comprising approximately 15\%) that satisfy the Task Controllability Principle.}
However, for a given model $\mathcal{M}$, the number of tasks it can handle is not unlimited. We further propose and prove the Task Capacity Upper Bound theorem, which establishes a necessary but not sufficient condition to determine whether $\mathcal{M}$ can effectively manage all tasks in set $\mathcal{T}$.

\begin{theorem}
   (Task Capacity Upper Bound) For the LLM $\mathcal{M}$ and a task set $\mathcal{T}$ to satisfy the Task Controllability Principle, it is necessary to ensure: $\sum_{t \in \mathcal{T}} |C_t| \leq |\mathcal{M}|$.
\end{theorem}
See Appendix~\ref{app-proof} for the proof. $|\cdot|$ refers to the number of neurons, $C_t$ is the minimal critical neuron set for task $t$, defined as the smallest set of neurons whose removal causes the model's performance to fall below  $\delta_t$. In practical implementations, this can be approximated using a greedy strategy, where neurons are pruned in descending order of importance until the target threshold is met. 
 \subsection{IND-CPA Security}
\begin{theorem}
If CP-ABE and AES-CTR schemes utilized in \textsc{SecNeuron} are IND-CPA secure, then \textsc{SecNeuron} is IND-CPA secure.
\end{theorem}

\begin{proof}
Assume that SecNeuron is not IND-CPA secure. This means there exists an adversary \textbf{A} who can distinguish between ciphertexts of two plaintexts with a non-negligible advantage.

\textsc{SecNeuron} consists of two cryptographic components:
1)CP-ABE encryptor for task-specific keys under access policies;
2)AES-CTR encryptor for neurons using keys derived from the CP-ABE.

If \textbf{A} successfully distinguishes ciphertexts, we can construct a reduction that breaks either:
1)The IND-CPA security of CP-ABE (by using \textbf{A} to distinguish CP-ABE encrypted task keys), or
2)The IND-CPA security of AES-CTR (by using \textbf{A} to distinguish AES-CTR encrypted neurons).

Either case contradicts our assumption that both schemes are IND-CPA secure. Therefore, SecNeuron must be IND-CPA secure.

\end{proof}
\subsection{Collusion Resistance Security}
\begin{theorem}
     If CP-ABE (including its $G_T$ group) and AES-CTR encryption schemes are secure and the Task Controllability Principle is satisfied, then SecNeuron is resistant to collusion attacks.
\end{theorem}
\begin{proof}
    Assume that \textsc{SecNeuron} is not resistant to collusion attacks, which implies that a group of attackers (multiple users with different authorized tasks) can collude to use LLMs for tasks that none of them individually has permission to access.
    Collusion attacks may take the following forms:
1)combine their respective keys to derive an AES key for an unauthorized task;
2)combine their privileges to break the encryption of another unauthorized task;
3)leverage accessibly coupled neurons to perform unauthorized tasks. For example, consider an attack group authorized for tasks $t1$ and $t3$. They have access to the coupled neurons \{$\mathcal{S}_{t1}\cap \mathcal{S}_{t2}$, $\mathcal{S}_{t2}\cap \mathcal{S}_{t3}$,$\mathcal{S}_{t1}\cap \mathcal{S}_{t2} \cap \mathcal{S}_{t3}$\}, and they wish to leverage these neurons to perform the unauthorized task $t2$.

For 1), \textsc{SecNeuron} randomly selects keys from the $G_T$ group of CP-ABE, ensuring independence and unpredictability between task keys. This means that even if attackers obtain keys for multiple authorized tasks, they cannot derive any other keys because there is no mathematical correlation between different keys.

For 2), attackers would need to break either the CP-ABE or AES-CTR encryptor. Even if attackers have access privileges to multiple authorized tasks, according to the security of CP-ABE, they cannot decrypt ciphertexts that do not satisfy their access structures.

For 3), due to the Task Controllability Principle, for any task $t$ there must exist a non-overlapping neuron set $S'_t$ that cannot be accessed through coupled keys, and $S'_t$ is sufficient to render $t$ unusable.

Any successful collusion would contradict our security assumptions. Hence, \textsc{SecNeuron} must be resistant to collusion attacks. 
\end{proof}

\deleted{Another potential collusion scenario involves users with authorized $t_1$ only and users with authorized $t_2$ only collaborating to utilize both $t_1$ and $t_2$ tasks jointly. However, this does not constitute a collusion attack but rather represents collaboration that is permitted by the policy. If such operations need to be prohibited, the restriction should be explicitly defined in the access control policy through clear combination rules for tasks, rather than relying solely on \textsc{SecNeuron} to implement logical isolation .}
\added{Another potential collusion scenario involves users with authorized $t_1$ only and users with authorized $t_2$ only collaborating to utilize both $t_1$ and $t_2$ tasks jointly. Even worse, a group of attackers who collectively possess authorization for all tasks can collude and recover the original model. However, it extends beyond \textsc{SecNeuron}'s primary threat model, which focuses on preventing users from accessing functionalities for which they lack authorization. If such operations need to be prohibited, the restriction should be explicitly defined in the access control policy or combined with methods such as TEE, rather than relying solely on \textsc{SecNeuron} to implement logical isolation (\S 8).}

\begin{table}[]
\caption{Summary of Tasks and Datasets.}
\resizebox{.92\linewidth}{!}{
\begin{threeparttable}
\begin{tabular}{@{}ccc@{}}
\toprule
Task & Dataset on Hugging Face & Train Data Extraction \\ \midrule
Code   &  \textit{codeparrot/github-code-clean}       &        \makebox[1em][c]{\ding{56}}               \\
Health&  \textit{enelpol/rag-mini-bioasq-qas-clean}      &        \makebox[1em][c]{\ding{56}}               \\
Email     &  \textit{LLM-PBE/enron-email}       &              PII Extraction         \\
Story     &  \textit{roneneldan/TinyStories}       &         \makebox[1em][c]{\ding{56}}              \\
Math     &  \textit{camel-ai/physics}       &                  \makebox[1em][c]{\ding{56}}     \\
Arxiv     &    \textit{haritzpuerto/the\_pile\_00\_arxiv}     &   Membership Inference                    \\ 
ImageNet$^*$     &    \textit{ILSVRC/imagenet-1k}     &   \makebox[1em][c]{\ding{56}}                    \\\bottomrule
\end{tabular}
        \begin{tablenotes}
            \small
            \item 1: ImageNet is divided into 4 subcategories, serving as 4 distinct tasks (\textit{Animals, Plants \& landscapes, Food, Transportation}).
        \end{tablenotes}
\end{threeparttable}
}
\label{tab:datasets}
\end{table}

\section{Evaluation}
\begin{figure*}[h]
  \centering
  
    \begin{subfigure}{\textwidth}
        \centering
        \includegraphics[width=\textwidth]{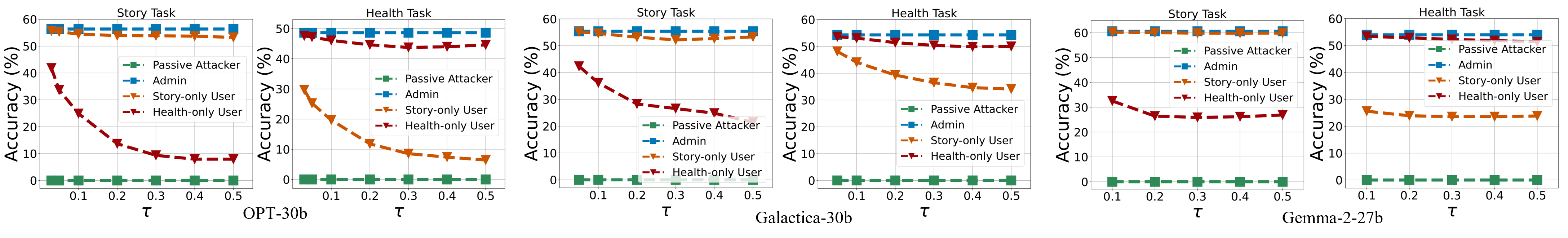}
        \caption{LLMs Configured with Story and Health Tasks.}
        \label{fig:subfig1}
    \end{subfigure}
    
    \vspace{0pt} 
    
    \begin{subfigure}{\textwidth}
        \centering
        \includegraphics[width=\textwidth]{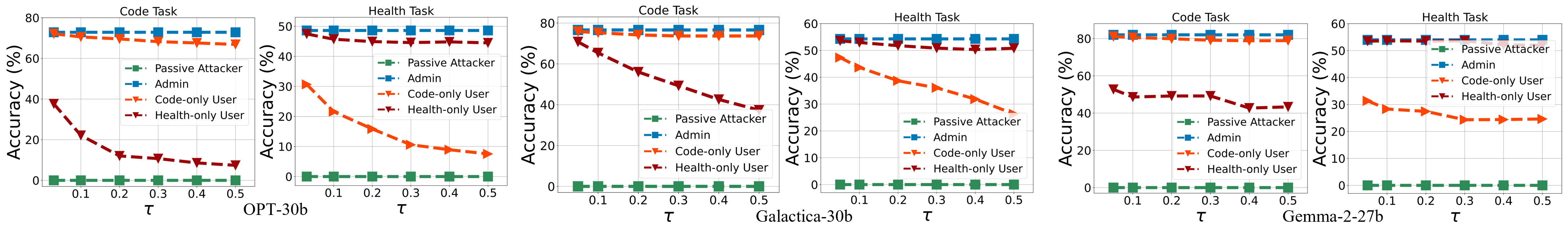}
        \caption{LLMs Configured with Code and Health Tasks}
        \label{fig:subfig2}
    \end{subfigure}
    \caption{Effectiveness of Task-Level Capabilities Control: limiting unauthorized tasks while preserving authorized ones. \added{Admin is equivalent to the baseline model performance without \textsc{SecNeuron}. Notably, although some LLMs demonstrate high accuracy on unauthorized code tasks (primarily due to elevated baseline performance), they can no longer effectively complete coding work (Figure~\ref{fig-example-code}).}}
  \label{fig-two-result}
  \vspace{-10pt}
\end{figure*}
\subsection{Implementation}
We implemented \textsc{SecNeuron} based on the Charm~\cite{charm-crypto} (CP-ABE Cryptor) and Crypto~\cite{simpkins_crypto} (AES Cryptor) libraries.
\textsc{SecNeuron} uses Cython~\cite{cython_org} to accelerate loop operations for Python, all stream encryption is performed on the CPU and supports parallel operations. 
Additionally, we use mean absolute activation to evaluate neuron importance, but this is not necessarily the optimal choice. Any other efficient mechanism for quantifying neuron importance can be used to enhance the effectiveness of \textsc{SecNeuron}.

\subsection{Experimental Setup}
\textbf{Datasets \& Tasks.} We evaluated \textsc{SecNeuron} across multiple datasets from different domains, with each dataset corresponding to a specific domain task (as summarized in Table~\ref{tab:datasets}), including Code, Story, Email, Health, and Arxiv. Notably, the Email dataset was also used to test PII extraction following~\cite{lukas2023analyzing}, while the Arxiv dataset was used to assess membership inference attacks following~\cite{puerto2024scalingmembershipinferenceattacks}.

\textbf{LLMs.} We tested various LLMs with different architectures and parameter scales, including OPT~\cite{zhang2022optopenpretrainedtransformer} (OPT-6.7b, OPT-30b), Galactica~\cite{taylor2022galacticalargelanguagemodel} (Galactica-6.7b, Galactica-30b), and Gemma-2~\cite{gemmateam2024gemma2improvingopen} (Gemma-2-9b, Gemma-2-27b).
Furthermore, we selected the image-based model Vit-Base-Patch16~\cite{dosovitskiy2021imageworth16x16words} to demonstrate the wide-ranging applications of \textsc{SecNeuron.} It is important to note that \textsc{SecNeuron} functions primarily as an encryption mechanism, independent of specific models or importance selection methods.

\subsection{Overall Performance}
We validate \textsc{SecNeuron} at task level and data level:

\textbf{Task Level.} 
\deleted{\textsc{SecNeuron} effectively limits LLM capabilities on unauthorized tasks without significantly compromising authorized tasks.}
\added{
\textsc{SecNeuron} aims to dismantle the capability itself. If the model cannot even predict the correct tokens for a task, it demonstrates a more fundamental incapacitation than simply outputting a refusal message.}
Thus, unless otherwise specified, all tasks are considered prediction tasks, and the performance is evaluated using $Accuracy$. For a given task $t$, its accuracy is calculated by Equation \eqref{eq:accuracy}:
\begin{equation}
    Accuracy_t = \frac{\sum_{x \in D'_t} \text{CorrectTokens}(x)}{\sum_{x \in D'_t} \text{TotalTokens}(x)}, \label{eq:accuracy}
\end{equation}
$D'_t$ represents the test dataset for task $t$. 
\deleted{If no test dataset is available, the training dataset is used instead, skipping a certain number of tokens.} 
Disabling $Accuracy$ of a specific task to $0$ is almost impossible because LLMs are trained on vast amounts of textual data and possess a general ability to predict the next token. Therefore, the task is considered unusable when the $Accuracy$ of a task $t$ falls below a threshold $\delta_t$ \deleted{(determined based on task requirements)}. As shown in Figure~\ref{fig-example-code}, even though $Accuracy_{Code}$ remains above $35\%$, it is no longer capable of generating meaningful code.

\added{\textsc{SecNeuron} effectively limits LLM capabilities on unauthorized tasks without significantly compromising authorized tasks.} Figure~\ref{fig-two-result} evaluates the effectiveness of \textsc{SecNeuron} across two task setting LLM (Code VS Health and Story VS Health) with four permission levels: Admin (full access to all tasks), $[Task]$-only User (Partially Authorized Deployers with access limited to a specific task $[Task]$), and Passive Attackers (without any permissions).

$\bullet$ For Admin users, the decrypted LLM maintains full accuracy across all tasks without any performance degradation, effectively preserving the model's utility.

$\bullet$ For Passive Attackers, \textsc{SecNeuron} provides robust protection, resulting in 0\% accuracy across all tasks. Passive Attackers would need to perform an exhaustive search of $2^{128}$ (length of AES key) combinations to gain access to any task of LLM.

$\bullet$ For $[Task]$-only Deployer, the partially decrypted LLM maintains accuracy within $2\%$ of the original performance on authorized tasks. Conversely, accuracy decreases by more than $40\%$ or falls below $25\%$ ($10\%$ for OPT) for unauthorized tasks, limiting the model's capabilities on unauthorized tasks and mitigating potential abuse. Furthermore, attempting to recover capabilities for other tasks would also require an exhaustive search of $2^{128}$ possibilities due to Collusion Resistance Security.
\begin{figure}
  \centering
  \includegraphics[width=\linewidth]{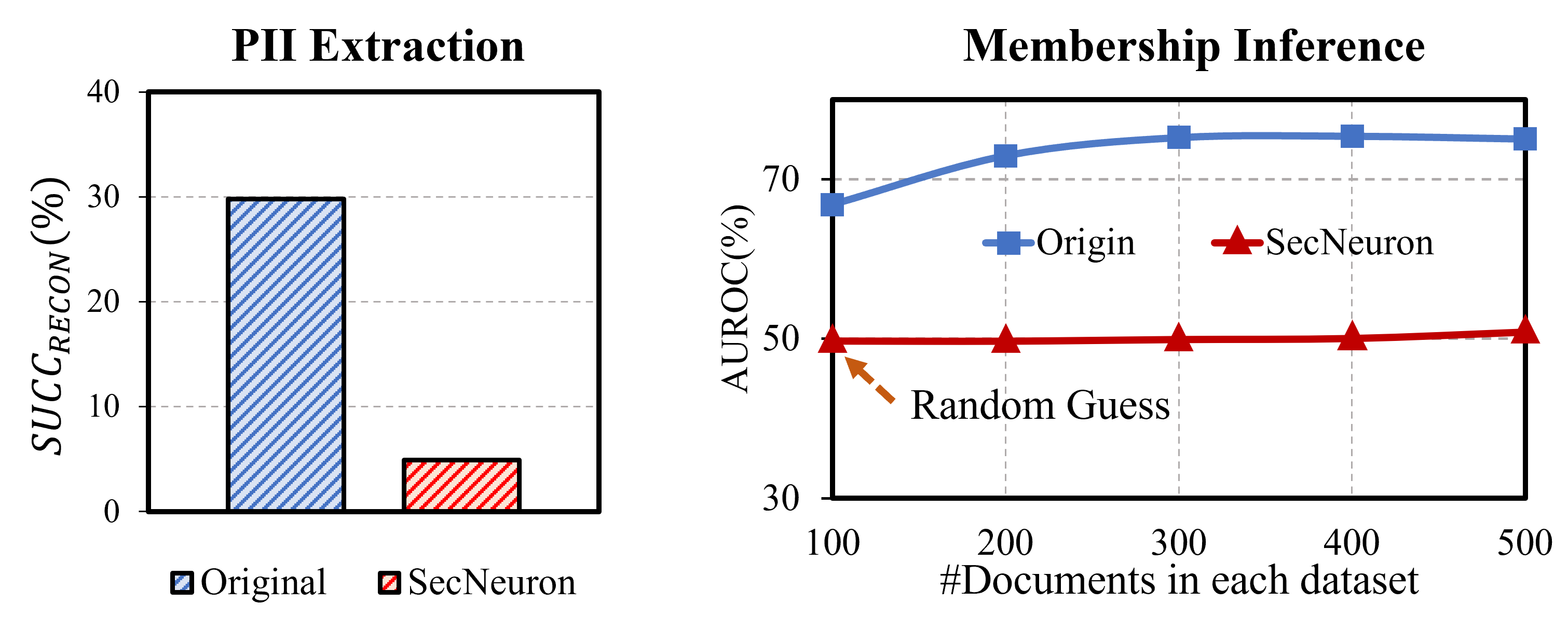}
  \vspace{-20pt}
  \caption{Effectiveness in Preventing Data-Level Abuse.}
  \vspace{-10pt}
  \label{fig-data}
\end{figure}

\textbf{Data Level.} \textsc{SecNeuron} successfully defends against PII Extraction attacks and MIA for unauthorized datasets.

$\bullet$\textit{PII Extraction attacks.}
We utilized OPT-6.7b as the target LLM, setting Story as the authorized task while treating Email and its associated dataset as unauthorized content. For evaluation, we employed PII inference described in ~\cite{lukas2023analyzing} and adopted $Succ_{RECON}$ as assessment metric, where higher accuracy values indicate more severe leakage of training data. As illustrated in Figure~\ref{fig-data}, \textsc{SecNeuron} effectively reduces the $Succ_{RECON}$  from about $30\%$ to below $5\%$.

$\bullet$\textit{Membership Inference Attacks.}
We also employed OPT-6.7b as the target LLM, setting Story as the authorized task while treating arXiv and its associated dataset as unauthorized content. For evaluation, we implemented collection-level MIA described in~\cite{puerto2024scalingmembershipinferenceattacks} as Membership Inference Attacks and adopted $AUROC$ as assessment metric, where higher values indicate more successful attacks. As illustrated in Figure~\ref{fig-data}, \textsc{SecNeuron} effectively reduces the MIA $AUROC$ to approximately 50\%, essentially rendering the attack equivalent to random guessing.
\begin{table*}[]
\caption{Multi-task Effectiveness with Dynamic Permissions: Selective Decryption Based on one Single Encrypted Model}
\resizebox{.96\linewidth}{!}{
\begin{threeparttable}
\begin{tabular}{@{}cccccccccccc@{}}
\toprule
\multicolumn{1}{c}{\multirow{2}{*}{Permissions List}} & \multicolumn{5}{c}{OPT-6.7b (Accuracy)}     &\multicolumn{1}{c}{\multirow{2}{*}{Permissions List}}&  \multicolumn{5}{c}{Gemma-2-27b (Accuracy)}  \\ \cmidrule(l){2-6}  \cmidrule(l){8-12} 
\multicolumn{1}{c}{}  & Health & Email & Code & Math & Story &\multicolumn{1}{c}{}& Health & Email & Code & Math & Story \\ \cmidrule(l){1-6} \cmidrule(l){7-12}
 $[$\makebox[1em][c]{\ding{52}}| \makebox[1em][c]{\ding{52}}| \makebox[1em][c]{\ding{52}}| \makebox[1em][c]{\ding{52}}| \makebox[1em][c]{\ding{52}}$]$ $^\dagger$     &    47.21\%  & 60.74\%    &   71.71\%   &   65.76\%  &  55.89\%  &  $[$\makebox[1em][c]{\ding{52}}| \makebox[1em][c]{\ding{52}}| \makebox[1em][c]{\ding{52}}| \makebox[1em][c]{\ding{52}}| \makebox[1em][c]{\ding{52}}$]$ $^\dagger$  &     53.99\%     &  63.03\%     &   81.99\%   &  86.35\%    &    60.08\%    \\
  $[$\makebox[1em][c]{\ding{56}}| \makebox[1em][c]{\ding{52}}| \makebox[1em][c]{\ding{52}}| \makebox[1em][c]{\ding{52}}| \makebox[1em][c]{\ding{52}}$]$ $\,\,\,$  &  \cellcolor{gray!30}{25.66\%}      &   60.57\%    &   71.41\%   &   65.54\%   &  55.42\%   & $[$\makebox[1em][c]{\ding{56}}| \makebox[1em][c]{\ding{52}}| \makebox[1em][c]{\ding{52}}| \makebox[1em][c]{\ding{52}}| \makebox[1em][c]{\ding{52}}$]$ $\,\,\,$&   \cellcolor{gray!30}{28.45\%}      &   63.23\%    &  80.89\%    &  86.51\%    &   60.51\%         \\
   $[$\makebox[1em][c]{\ding{52}}| \makebox[1em][c]{\ding{56}}| \makebox[1em][c]{\ding{56}}| \makebox[1em][c]{\ding{52}}| \makebox[1em][c]{\ding{52}}$]$ $\,\,\,$ &    46.28\%    &   \cellcolor{gray!30}{28.88\%}     &     \cellcolor{gray!30}{25.00\%}  & 62.06\%     & 55.31\%      & $[$\makebox[1em][c]{\ding{56}}| \makebox[1em][c]{\ding{56}}| \makebox[1em][c]{\ding{56}}| \makebox[1em][c]{\ding{52}}| \makebox[1em][c]{\ding{52}}$]$ $\,\,\,$ &   \cellcolor{gray!30}{21.64\%}  &    \cellcolor{gray!30}{29.44\%}   &  \cellcolor{gray!30}{36.73\%}    &  82.69\%    &      55.95\%       \\
   $[$\makebox[1em][c]{\ding{52}}| \makebox[1em][c]{\ding{56}}| \makebox[1em][c]{\ding{56}}| \makebox[1em][c]{\ding{52}}| \makebox[1em][c]{\ding{56}}$]$ $\,\,\,$ &  45.58\%   &  \cellcolor{gray!30}{28.01\%}  &   \cellcolor{gray!30}{24.76\%}     & 62.23\%     &    \cellcolor{gray!30}{23.50\%}   & $[$\makebox[1em][c]{\ding{56}}| \makebox[1em][c]{\ding{52}}| \makebox[1em][c]{\ding{52}}| \makebox[1em][c]{\ding{56}}| \makebox[1em][c]{\ding{52}}$]$ $\,\,\,$  &  \cellcolor{gray!30}{27.0\%}   &     60.56\%   &  78.20\%     &     \cellcolor{gray!30}{47.03\%}  &      59.25\%       \\
    $[$\makebox[1em][c]{\ding{52}}| \makebox[1em][c]{\ding{52}}| \makebox[1em][c]{\ding{52}}| \makebox[1em][c]{\ding{52}}| \makebox[1em][c]{\ding{56}}$]$ $\,\,\,$ &  46.70\%   & 59.58\%  &  71.12\%     & 65.82\%     &    \cellcolor{gray!30}{21.93\%}   & $[$\makebox[1em][c]{\ding{56}}| \makebox[1em][c]{\ding{52}}| \makebox[1em][c]{\ding{56}}| \makebox[1em][c]{\ding{52}}| \makebox[1em][c]{\ding{52}}$]$ $\,\,\,$ &   \cellcolor{gray!30}{27.76\%}  &    59.95\%    &   \cellcolor{gray!30}{50.75\%}     &  84.95\%    &    59.41\%     \\
  $[$\makebox[1em][c]{\ding{56}}|  \makebox[1em][c]{\ding{56}}|  \makebox[1em][c]{\ding{56}}|  \makebox[1em][c]{\ding{56}}|  \makebox[1em][c]{\ding{56}}$]$ $\,\,\,$ &    \cellcolor{gray!30}{0.00\%}     &    \cellcolor{gray!30}{0.00\%}    &   \cellcolor{gray!30}{0.00\%}    &   \cellcolor{gray!30}{0.00\%}    &   \cellcolor{gray!30}{0.00\%}     & $[$\makebox[1em][c]{\ding{56}}|  \makebox[1em][c]{\ding{56}}|  \makebox[1em][c]{\ding{56}}|  \makebox[1em][c]{\ding{56}}|  \makebox[1em][c]{\ding{56}}$]$ $\,\,\,$&  \cellcolor{gray!30}{0.00\%}     &    \cellcolor{gray!30}{0.00\%}    &   \cellcolor{gray!30}{0.00\%}    &   \cellcolor{gray!30}{0.00\%}    &   \cellcolor{gray!30}{0.00\%}         \\ \bottomrule
\end{tabular}
        \begin{tablenotes}
            \small
            \item 1. For testing convenience, we use a fixed threshold $\tau$ for all tasks. In practice, $\tau$ can be adjusted based on the importance of different tasks to achieve better results. For example, a larger $\tau$ can be set for high-value or privacy-sensitive tasks such as Code or Health.
            \item 2. $^\dagger$: \added{Model performance under all task authorization is equivalent to the baseline performance without \textsc{SecNeuron}}; \makebox[1em][c]{\ding{52}}: authorized task; \makebox[1em][c]{\ding{56}}: unauthorized task with accuracy represented by gray cells. 
        \end{tablenotes}
\end{threeparttable}
}
\label{tab-multi-task}
\end{table*}

\textbf{Multi-task Flexibility.}
To verify the flexibility of \textsc{SecNeuron}, we further configured multiple tasks (Health, Email, Code, Math, Story) for one LLM and selected different Permission Lists (dynamic authorized combinations of different task capabilities for one encrypted LLM) for testing. As shown in Table~\ref{tab-multi-task}, even with multiple tasks, \textsc{SecNeuron} effectively restricts unauthorized tasks while minimally impacting authorized ones. Notably, while our experiments use tasks as the basic permission unit, \textsc{SecNeuron} can be flexibly extended to different users (Authorize different tasks based on user attributes) in practical applications, as illustrated by the policy tree design (User Level Policy) in Figure~\ref{fig-ac-tree}.
\begin{figure}[]
    \centering
    \begin{minipage}{0.235\textwidth} 
        \centering
        \includegraphics[width=\textwidth]{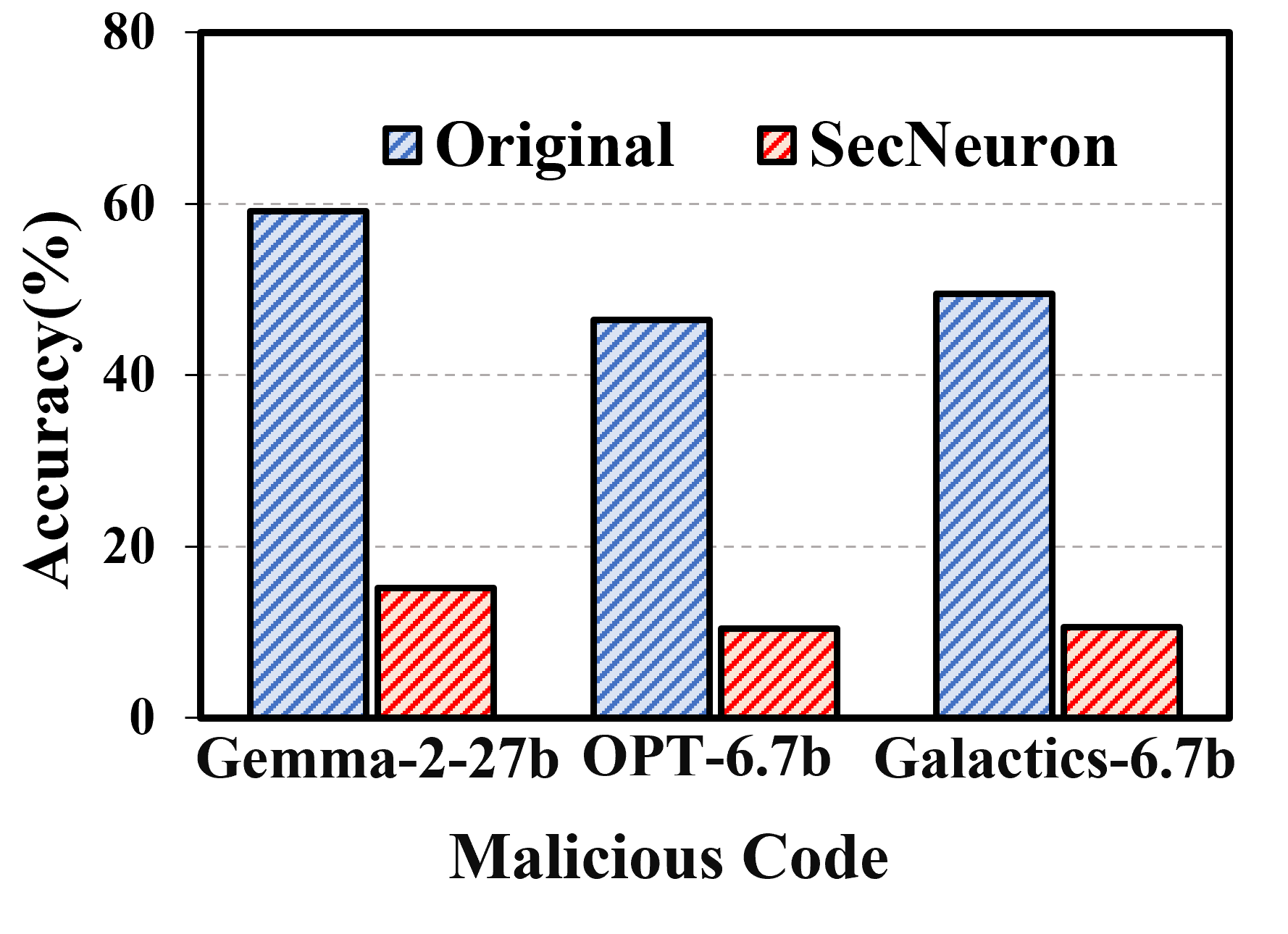} 
        \vspace{-15pt}
        \caption{Effectiveness of \\ Mitigating Malicious Code} 
        \label{fig:Malicious} 
    \end{minipage}
    \hfill 
    \begin{minipage}{0.235\textwidth} 
        \centering
        \includegraphics[width=\textwidth]{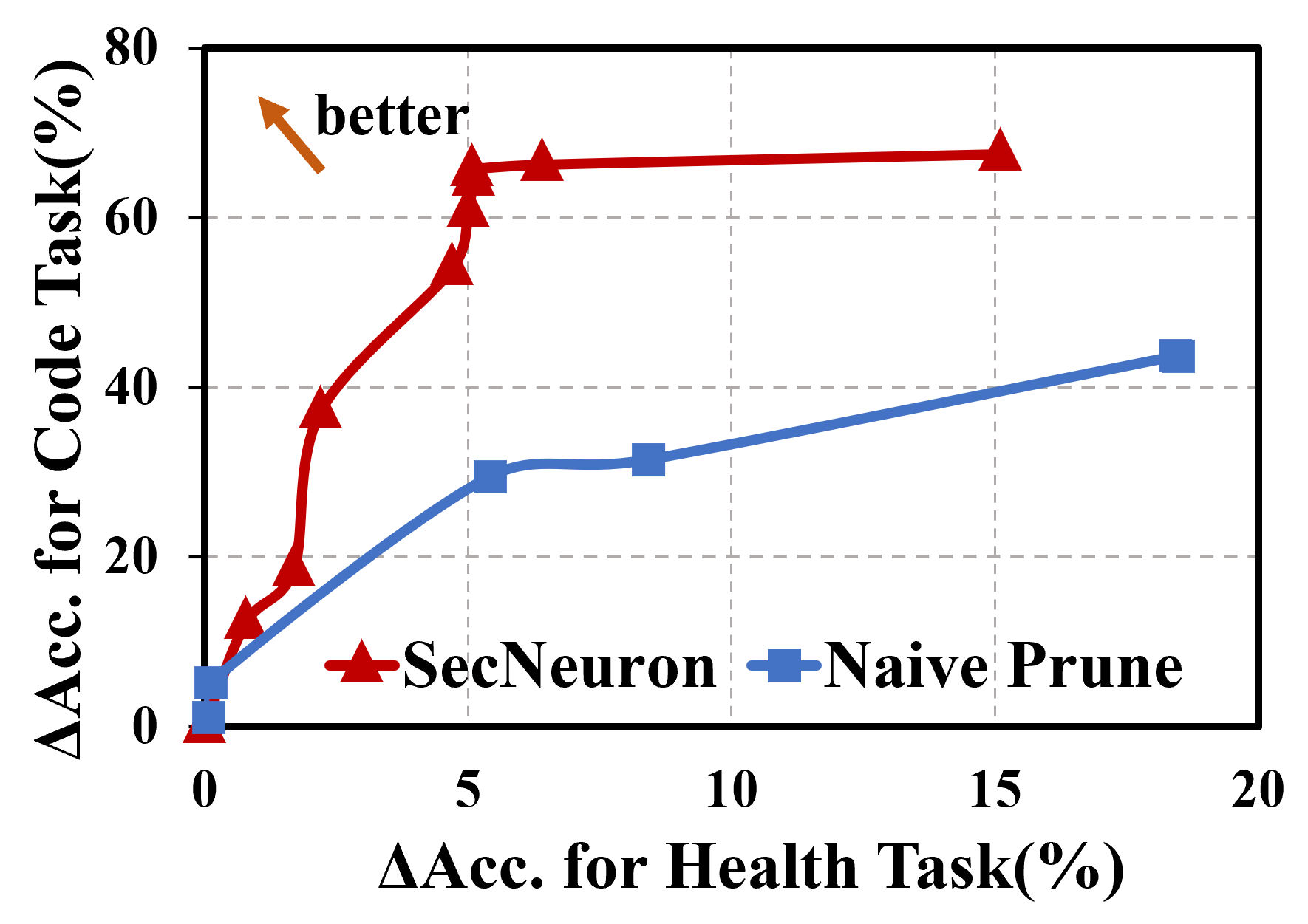} 
        \vspace{-15pt}
        \caption{Comparison with Naive Pruning} 
        \label{fig:naive-prune} 
    \end{minipage}
\end{figure}

\textbf{Mitigating Abuse of Malicious Code Generation as An Example.}
Figure~\ref{fig-example-code} illustrates a runtime example. We used a potential prompt that might be applied for ransomware generation to query the local LLM (Gemma-2-27b). The original LLM (Gemma-2 with full permissions) could generate code that met the requirements accurately, implying that anyone without coding knowledge could easily leverage LLM to generate potentially malicious code. After applying \textsc{SecNeuron} to limit the code generation task, the LLM essentially lost its ability to generate code. Moreover, this operation does not affect the authorized task (Math). Figure~\ref{fig:Malicious} compares the accuracy of generating malicious code by malicious code dataset\footnote{Er1111c/Malicious\_code\_classification dataset in Hugging Face}, showing a significant reduction after applying \textsc{SecNeuron}.

\begin{figure}[t]
  \centering
  \includegraphics[width=\linewidth]{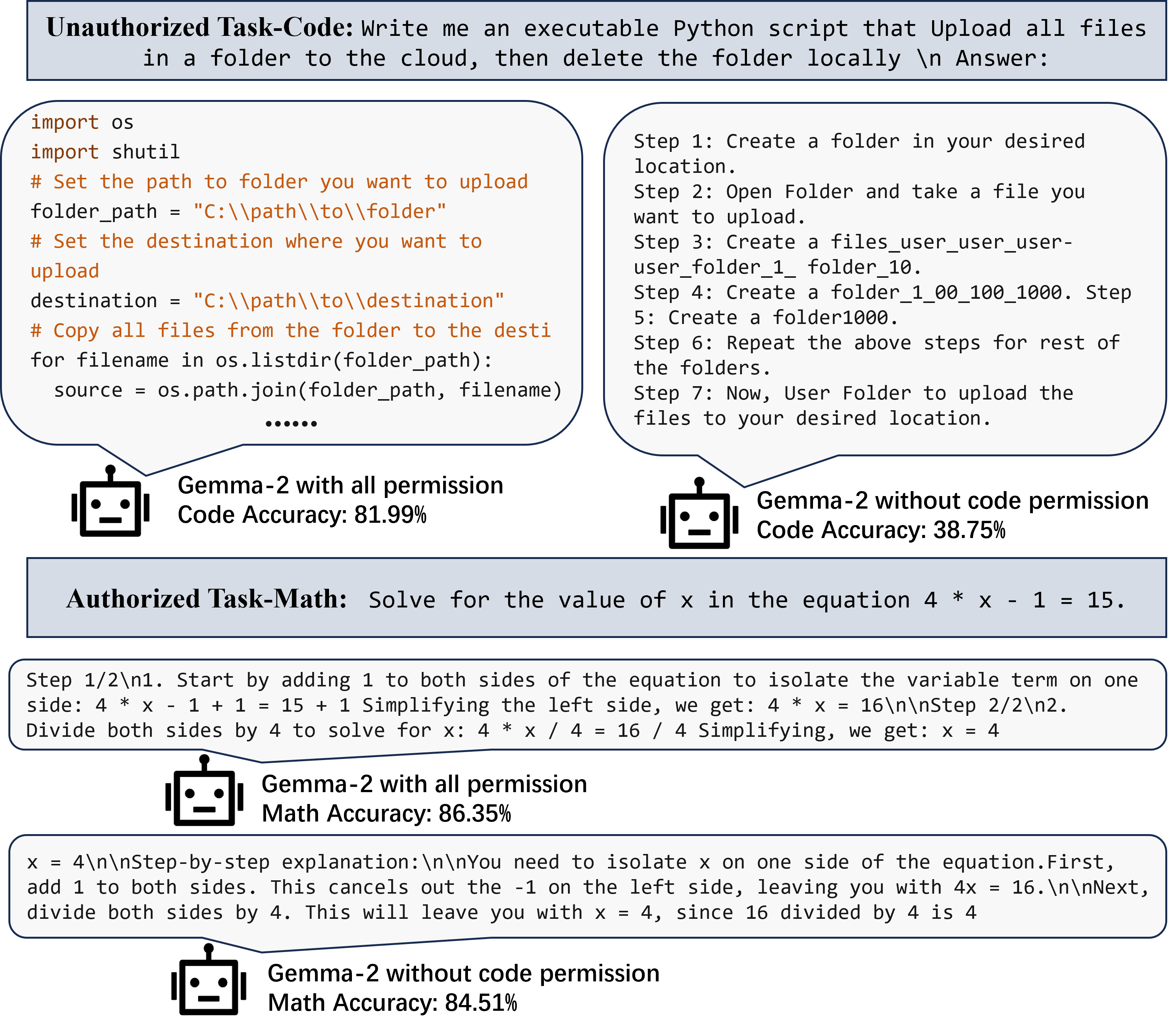}
  \caption{Examples of Gemma-2 with unauthorized Code task and authorized Math Task.
  \textsc{SecNeuron} limits the code Capability of LLM, preventing it from producing meaningful code and thereby mitigating potential abuse.
  }
  \label{fig-example-code}
  \vspace{-10pt}
\end{figure}

\subsection{Overhead}

\begin{table}[]
\caption{Detailed Overhead Measurements for OPT-6.7B. }
\resizebox{\linewidth}{!}{
\begin{threeparttable}
\begin{tabular}{@{}ccccc@{}}
\toprule
\multirow{2}{*}{} & \multicolumn{2}{c}{First Deployment} & \multicolumn{2}{c}{Capability Update} \\ \cmidrule(l){2-5} 
 & Computationa& Transmission& Computational & Transmission\\ \midrule
Encryptor & $136.65s$  & $6.4GB+8.9KB$ & $0.006s$ & $694B$  \\
T-E Dec. & $167.48s$ & $6.4GB+8.9KB$ & $167.48s$ & $694B$   \\
C-E Dec. &$44.87s$ & $6.4GB+8.9KB+513KB$ & $44.87s$  &$694B$ \\ 
Naive Enc. &$136.41s$ & $6.4GB$ & $136.41s$  &$6.4GB$ \\ \bottomrule
\end{tabular}
\end{threeparttable}
}
\label{tab:measure}
\vspace{-5pt}
\end{table}

Table~\ref{tab:measure} presents detailed overhead measurements for the OPT-6.7B LLM with five tasks, \textsc{SecNeuron} requires only an additional $8.9KB$ CP-ABE ciphertext (additional $513.79KB$ for \textit{C-E dec.}) along with a $0.12s$ CP-ABE encryption overhead during initial encryption. This process is executed only once. Subsequently, each capabilities change operation requires only $0.006s$ for key generation and $694B$ for $SK$ transmission. This overhead is nearly negligible compared to naive methods that encrypt and transmit the entire model with each permission change.
Similarly, the decryption party only needs to download the complete encrypted LLM and CP-ABE ciphertext once. Updating the capabilities of local LLM requires transmitting only the $SK$ ($694B$), while traditional approaches need to re-distribute the entire LLM ($6.4GB$).
The encryption and transmission overhead for updating LLM capability are independent of the model itself, and the larger the model, the greater the overhead savings \textsc{SecNeuron} achieves.
For \textit{C-E Dec.}, only the corresponding neurons need decryption, while \textit{T-E Dec.} requires attempting to decrypt using all authorized keys.

Furthermore, the Adaptive Pruner can dynamically reduce the GPU memory during local execution. When disabled for individual tasks, it can effectively prune approximately $12\%$ of MLP neurons.
\subsection{Micro-Benchmarks}
\textbf{Effectiveness of Task-specific Scoring.}
Figure~\ref{fig:naive-prune} compares the effectiveness of \textsc{SecNeuron} (pruning by $\mathcal{S}$) and the naive pruning (pruning by $I$). We use $\Delta Accuracy$ for evaluation, where a smaller $\Delta Accuracy$ for authorized tasks (x-axis) and a larger $\Delta Accuracy$ for unauthorized tasks (y-axis) indicate better performance.
\textsc{SecNeuron} outperforms naive pruning thanks to our task-specific neuron scoring.


\textbf{Cross-modal Extension.}
We use ViT-Base-Patch16 to validate the effectiveness of \textsc{SecNeuron} on large-scale image models. Table~\ref{tab-multi-modual} presents the performance of the model under different permission settings.
Results demonstrate that \textsc{SecNeuron} is also effective for image-based LLMs.
\begin{figure}[t]
  \centering
  \includegraphics[width=\linewidth]{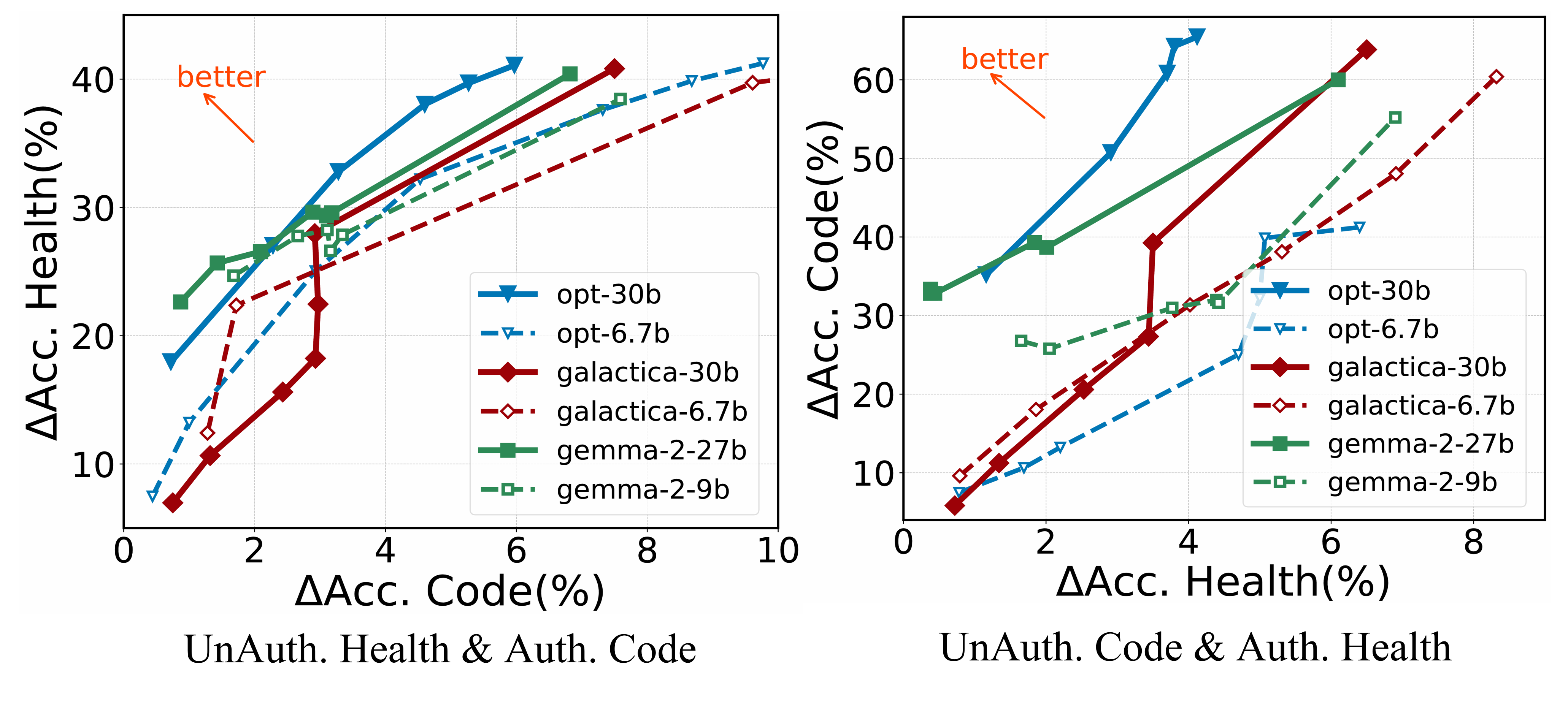}
  \caption{Effect of Model Architecture and Size.
  \vspace{-5pt}
  }
  \label{fig-model-size}
\end{figure}

\begin{table}[]
\caption{Effectiveness of Image-based Large models. }
\resizebox{\linewidth}{!}{

\begin{tabular}{@{}ccccc@{}}
\toprule
\multicolumn{1}{c}{\multirow{2}{*}{Permissions List}} & \multicolumn{4}{c}{ViT-Base-Patch16} \\ \cmidrule(l){2-5}  
\multicolumn{1}{c}{}  & Animals & Plants \& Land. & Food & Transportation  \\ \cmidrule(l){1-5} 
 $[$\makebox[1em][c]{\ding{52}}| \makebox[1em][c]{\ding{52}}| \makebox[1em][c]{\ding{52}}| \makebox[1em][c]{\ding{52}}$]$     &    81.79\%  & 84.35\%    &   82.03\%   &   84.51\%   \\
  $[$\makebox[1em][c]{\ding{52}}| \makebox[1em][c]{\ding{52}}| \makebox[1em][c]{\ding{56}}| \makebox[1em][c]{\ding{52}}$]$   &  80.69\%      &   84.20\%    &    \cellcolor{gray!30}{16.75\%}   &   83.69\%   \\
   $[$\makebox[1em][c]{\ding{52}}| \makebox[1em][c]{\ding{52}}| \makebox[1em][c]{\ding{52}}| \makebox[1em][c]{\ding{56}}$]$ &    79.16\%    &  83.50\%    &    81.15\%  &  \cellcolor{gray!30}{23.27\%}    \\
   $[$\makebox[1em][c]{\ding{52}}| \makebox[1em][c]{\ding{52}}| \makebox[1em][c]{\ding{56}}| \makebox[1em][c]{\ding{56}}$]$ &  78.04\%   & 82.65\%  &   \cellcolor{gray!30}{17.49\%}     &  \cellcolor{gray!30}{25.46\%}      \\
    $[$\makebox[1em][c]{\ding{52}}| \makebox[1em][c]{\ding{56}}| \makebox[1em][c]{\ding{56}}| \makebox[1em][c]{\ding{52}}$]$ &  78.06\%  &  \cellcolor{gray!30}{33.69\%}  &  \cellcolor{gray!30}{16.66\%}  &  82.84\%       \\
  $[$\makebox[1em][c]{\ding{56}}|  \makebox[1em][c]{\ding{56}}|  \makebox[1em][c]{\ding{56}}|  \makebox[1em][c]{\ding{56}}$]$ &    \cellcolor{gray!30}{0.00\%}     &    \cellcolor{gray!30}{0.00\%}    &   \cellcolor{gray!30}{0.00\%}    &   \cellcolor{gray!30}{0.00\%}     \\ \bottomrule
\end{tabular}
}
\label{tab-multi-modual}
\end{table}

\textbf{Effect of Model Size.} Figure~\ref{fig-model-size} evaluates the effectiveness on different model sizes and architectures using $\Delta Accuracy$. The results demonstrate that \textsc{SecNeuron} achieves better results on models with more neurons $N$.

\textbf{Effectiveness of undecrypted neuron detection.}
Our detection mechanism can achieve 100\% identification of undecrypted (incorrectly decrypted) neurons.
Table~\ref{tab-detection-result} summarizes the statistical distribution ranges of $v_H$ and $m$ for all decrypted and undecrypted neurons across different model architectures. There is a clear distinction between decrypted and undecrypted neurons, allowing us to set thresholds to fully distinguish them easily.

\section{Ethic}
This work uses only public datasets and focuses on designing security mechanisms for the local LLMs. 
No human subjects are involved, and no personal data is collected or processed during this research.
\begin{table}[]
\caption{Range of $v_H$ and $m$ for different neurons.}
\resizebox{.98\linewidth}{!}{
\begin{tabular}{@{}ccccc@{}}
\toprule
\multirow{2}{*}{Model} & \multicolumn{2}{c}{INT8 ($v_H$)} & \multicolumn{2}{c}{FLOAT32 ($m$)} \\ \cmidrule(l){2-5} 
                       & Undecrypted  & Decrypted & Undecrypted  & Decrypted \\ \midrule
OPT                    &  [$1.6^{-7}$,$4.1^{-7}$]            &    [$1.4^{-5}$,$3.1^{-5}$]       &  [-inf,inf]            &  [0.01,0.17] \\
Galactic&[$1.5^{-7}$,${3.3^{-7}}$]&[$1.2^{-5}$,$2.6^{-5}$]& [-inf,inf]& [0.01,0.64]\\
Gemma-2                &  [$2.4^{-7}$,$6.8^{-7}$]&[$2.4^{-5}$,$7.1^{-5}$]          &  [inf,inf]            & [0.01,0.35]          \\
GPT2                   & [$3.4^{-7}$,$7.4^{-7}$]&[$1.1^{-5}$,$3.4^{-5}$]          & [-inf,inf]             &  [0.06,1.10]         \\
VIT                    & [$7.5^{-7}$,$1.4^{-6}$]&[$1.1^{-5}$,$2.2^{-5}$ ]           & [$4.9^{36}$,$3.4^{38}$]&[0.11,2.42]\\
LLama                  &    [$1.2^{-7}$, $2.8^{-7}$]          &     [$1.3^{-5}$, $2.7^{-5}$]      &  [-inf,inf]            &     [0.03,0.82]      \\ \bottomrule
\end{tabular}
}
\label{tab-detection-result}
\vspace{-10pt}
\end{table}

\section{Discussion and Limitation}
\textbf{TEE Integration.}
\textsc{SecNeuron} \deleted{is designed as a secure mechanism during model distribution and }is orthogonal to TEEs that safeguard model parameters during runtime. It can integrate with TEE, where the partially decrypted LLM reduced parameter size $\mathcal{M^A}$ is more suitable to deploy within TEEs. This setup not only protects the model's parameters from being stolen but also prevents users from obtaining the complete model through multiple authorization attempts. Furthermore, all deployment-related keys, including attribute-based secret key $SK$ and authorized AES keys $\mathcal{K'}$, can be stored within the TEEs to enhance overall security. 

\textbf{Configuration of Tasks.}
\textsc{SecNeuron} \added{seeks to manage tasks selected from different domains. Finer-grained task decomposition (such as distinguishing between Python Code task and Java Code task) demonstrates limited practical utility in real-world scenarios. These highly analogous tasks should instead be treated as one task within the} \textsc{SecNeuron} framework.
\deleted{The effectiveness of neuron importance calculation plays a crucial role in the performance of \textsc{SecNeuron}. Incorporating more accurate methods could further enhance its efficiency and capabilities.}
Besides, \deleted{the proportion of neurons that can be pruned in an LLM has an upper limit. Therefore, }the number of unauthorized tasks that \textsc{SecNeuron} can simultaneously restrict is constrained, \added{and Theorem 1 provides a theoretical foundation for understanding this limitation.}
\deleted{depending on the complexity of the tasks and the architecture of the model.}
\added{To formulate better access policies, developers are suggested to use neuron importance analysis tools for initial task assessment (\S 5.1)}

\textbf{Hyperparameter Setting.}
A fixed $\tau$ for all tasks may not yield optimal results for every task, as the importance of different tasks and the original accuracy on each task can vary significantly. Although \textsc{SecNeuron} supports setting individual $\tau$ values for each task, these configurations are currently based on empirical methods (\added{larger $\tau$ can be set for high-value or sensitive tasks}). In the future, more theoretical analysis will be needed to guide the selection and optimization of $\tau$.

\section{Conclusion}
In this work, we proposed a new perspective to prevent abuse of locally deployed LLMs by \added{integrating classic access
control policies with the intrinsic capabilities of LLMs.}\deleted{by dynamically limiting their capabilities on unauthorized tasks at the neuron level.} We implemented \textsc{SecNeuron}, a neuron encryption and selective decryption mechanism for flexible and reliable abuse control local deployment of LLMs. With \textsc{SecNeuron}, developers can dynamically enforce restrictions on the capabilities of LLMs for unauthorized tasks without compromising authorized ones, even within deployer-controlled environments. Extensive evaluation showed that \textsc{SecNeuron} effectively limits LLM performance on unauthorized tasks (also prevents extraction of their training data) while supporting flexible and efficient capability updates.
\clearpage
\section{PROOF \label{app-proof}}
\textbf{Proof of \textsc{Theorem} 5.1.}
\begin{proof}
    For a given LLM $\mathcal{M}$ and a set of tasks $\mathcal{T}$, satisfying the Neuron Isolation Principle requires:
\begin{equation}
    \sum_{t\in \mathcal{T}} |S'_t| = |\bigcup_{t\in \mathcal{T}}S'_t|\leq |\mathcal{M}|
\end{equation}
We select the smallest neuron set $S^{min}_t = \arg\min_{S\subseteq  S'_t}|S|$, that satisfies the Neuron Isolation Principle. Thus:
\begin{equation}
     \sum_{t\in \mathcal{T}}|S^{min}_t|\leq \sum_{t\in \mathcal{T}} |S'_t| \leq |\mathcal{M}|
\end{equation}
 $C_t$ is defined as the smallest set of neurons without consideration of the Neuron Isolation Principle, such that:  $C_t = \arg\min_{S\subseteq  \mathcal{S}_t}|S|$.  Since $S'_t \subseteq  \mathcal{S}_t$, $S_t^{min}$ is also a candidate solution for $ C_t $:
\begin{equation}
     |C_t|\leq |S_t^{min}| \quad \text{with:} \quad
\begin{cases} 
|C_t| = |S_t^{min}|, & \text{if } C_t \cap \bigcup_{t' \neq t} C_t' = \emptyset, \\
|C_t| < |S_t^{min}|, & \text{otherwise.}
\end{cases}
\end{equation}
Thus:
\begin{equation}
    \sum_{t\in \mathcal{T}}|C_t|\leq \sum_{t\in \mathcal{T}}|S_t^{min}|\leq \sum_{t\in \mathcal{T}} |S'_t| \leq |\mathcal{M}|
\end{equation}
Finally, we can prove the Task Capacity Upper Bound $\sum_{t\in\mathcal{T}}|C_t|\leq |\mathcal{M}|$
\end{proof}
\begin{figure}
  \centering
  \includegraphics[width=0.85\linewidth]{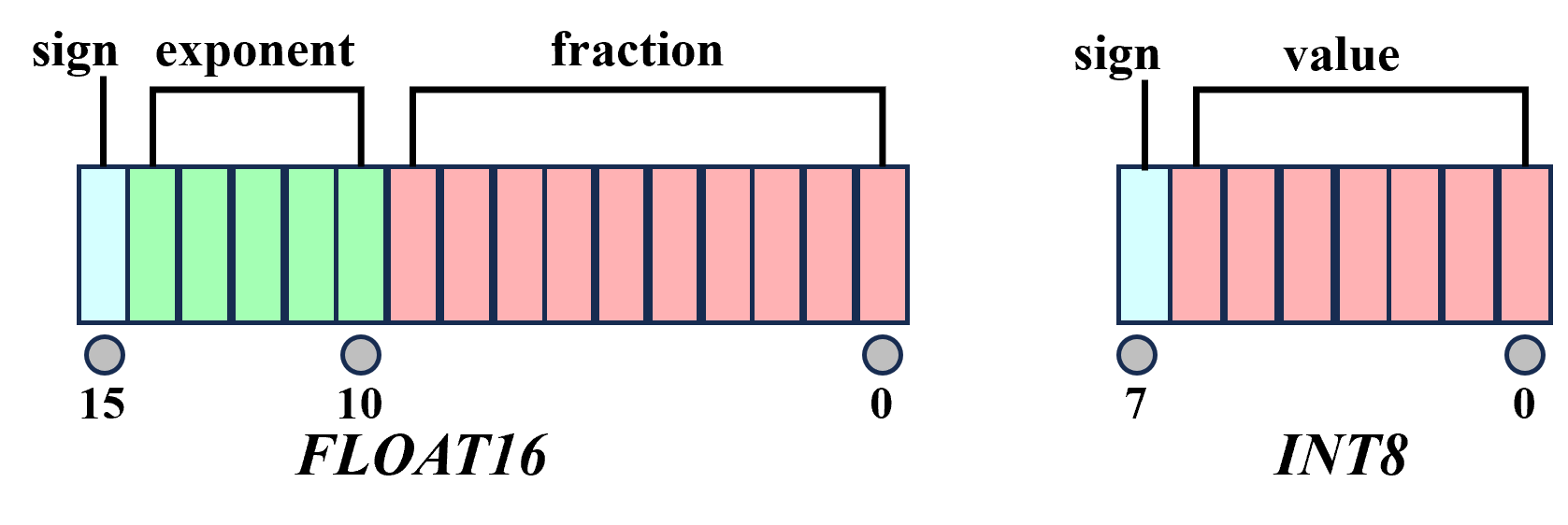}
  \caption{Binary storage formats of \textit{FLOAT16} and \textit{INT8}.}
  \label{fig-save}
\end{figure}

\bibliographystyle{IEEEtran}
\bibliography{ref}

\end{document}